\documentclass[12pt]{article}
\usepackage{titlesec}

\titleformat*{\section}{\large\bfseries}
\titleformat*{\subsection}{\normalsize\bfseries}
\usepackage{paralist}
\usepackage[dvipsnames]{xcolor}
\usepackage{tkz-euclide}
\usepackage[english]{babel}
\usepackage[utf8]{inputenc}
\usepackage{amsmath}
\usepackage{graphicx}
\usepackage[colorinlistoftodos]{todonotes}
\usepackage{amsfonts}
\usepackage[margin = 1.2 in]{geometry}
\usepackage{amsthm}
\usepackage{setspace}
\usepackage{enumerate}
\usepackage[multiple, hang, flushmargin]{footmisc}
\usepackage[hyperfootnotes=false]{hyperref}

\usepackage{amssymb}

\usepackage{mathrsfs}

\usepackage[citestyle=authoryear-comp, style = authoryear,natbib, maxcitenames=99, maxnames=99, sorting=nyt, uniquename = false, doi = false, giveninits = true]{biblatex}
\renewbibmacro{in:}{}
\DeclareFieldFormat{pages}{#1}
\renewbibmacro*{journal+issuetitle}{%
  \usebibmacro{journal}%
  \setunit*{\addcomma\addspace}%
  \iffieldundef{series}
    {}
    {\newunit
     \printfield{series}%
     \setunit{\addspace}}%
  \usebibmacro{volume+number+eid}%
  \setunit{\addspace}
  \usebibmacro{issue+date}%
  \setunit{\addcolon\space}%
  \usebibmacro{issue}%
  \newunit}
\AtEveryBibitem{
    \clearfield{issue}
    \clearfield{number}
    \clearfield{issn}
}
\uspunctuation
\usepackage{bbm}

\usepackage{float}
\usepackage{caption}

\usepackage{titling}

\usepackage{tikz,pgfplots}
\usetikzlibrary{patterns}
\pgfplotsset{compat = 1.10}
\usepgfplotslibrary{fillbetween}

\usepackage{enumitem}

\newtheorem{claim}{Claim}

\newtheorem{definition}{Definition}
\newtheorem{lemma}{Lemma}
\newtheorem{proposition}{Proposition}

\usepackage{bm}

\newcounter{thmnum}
\newcounter{lemnum}
\newcounter{clmnum}
\newcounter{exnum}
\begin{document}

\title{\Large Observability, Dominance,  and Induction in Learning Models\thanks{We thank Amanda Friedenberg, Ying Gao, and George Mailath  for helpful conversations, Jonathan Brownrigg for research assistance, and NSF grant 1951056 for financial support.}}
\author{Daniel Clark,\thanks{Department of Economics, MIT. Email: \href{mailto:dgclark@mit.edu}{dgclark@mit.edu}} ~ Drew Fudenberg,\thanks{Department of Economics, MIT. Email: \href{mailto:drew.fudenberg@gmail.com}{drew.fudenberg@gmail.com}} ~ and Kevin He\thanks{Department of Economics, University of Pennsylvania: \href{mailto:hesichao@gmail.com}{hesichao@gmail.com}}}
\date{January 3, 2022}

\maketitle
\thispagestyle{empty}
\vspace{-2em}

\begin{abstract}
\begin{singlespace}


Learning models do not in general imply that weakly dominated strategies are irrelevant or justify the related concept of ``forward induction,'' because rational agents may use dominated strategies as experiments to learn how opponents play, and may not have enough data to rule out a strategy that opponents never use. Learning models also do not support the idea that the selected equilibria should only depend on a game's  normal form, even though two games with the same  normal form present players with the same decision problems given fixed beliefs about how others play.  However, playing the extensive form of a game is equivalent to playing the  normal form augmented with the appropriate \emph{terminal node partitions} so that two games are \emph{information equivalent}, i.e., the players receive the same feedback about others' strategies.

 \end{singlespace}
\end{abstract}
\noindent \textbf{Keywords:} learning in games, equilibrium refinements, iterated dominance, forward induction
\newpage

\setcounter{page}{1}

\section{Introduction}
 The learning in games literature asks which equilibria are likely to persist in environments where new players are initially uncertain about the prevailing strategies  and learn about the strategy distribution by repeatedly playing the game.   One reason for the success of strategic stability and associated refinements is that they select intuitive equilibria in  signaling games, as shown by \citet{Banks1987} and \citet{Cho1987}. Learning models make similar predictions in these games (\citet{Fudenberg2018b}, \citet{fudenberg2020payoff},  \citet{clark2021justified}).  In this paper, we show  that these two sorts of refinements can have very different predictions in other games. Specifically, we show  that  learning models do not in general support either the iterated deletion of weakly dominated strategies or the related concept of ``forward induction.''\footnote{There are many related definitions of forward induction in the literature, see the papers surveyed in \citet{govindan2009forward}. As far as we know, none of these definitions has been accompanied by  a theory of how players come to have equilibrium beliefs, or why they should maintain their beliefs in the equilibrium or in others' rationality after observing a deviation.} We also show that learning models support only some  of the invariance axioms proposed by \citet{Kohlberg1986} and \citet{elmes1994strategic}.

There are two distinct reasons that the outcomes of learning models need not satisfy forward induction or iterated weak dominance. First, a dominated strategy may be used as an experiment to gain information about opponents' play at some off-path information sets, and the opponents may then correctly believe that the rare deviations from the equilibrium path use this dominated strategy. Second, even if a dominated strategy is never used, agents in other player roles may not learn this if they start with a prior belief to the contrary and don't obtain enough data to learn the truth.


The \citet{Kohlberg1986} argument that a solution concept for games should only depend  on the  normal form  is based on the claim that the differences between extensive forms with the same normal form are ``irrelevant details'' because they do not change the decision problem of a player who faces the \emph{same fixed and known} strategies of the opponents.  Because the  normal form abstracts from many aspects of game play that are relevant for how people  learn what strategies are used by others, there is no reason to expect learning to depend only on this very abstract representation of strategic interaction. Instead, the set of learning outcomes is only invariant to transformations that  are both \emph{decision invariant}, i.e.,  lead to the same best responses as a function of opponent strategies, and \emph{information invariant} in the sense of providing the same feedback to the agents in their learning problems.   Specifically, learning outcomes, unlike sequential equilibria, are invariant to the coalescing of consecutive moves by the same player. However, like sequential equilibria and unlike the various definitions of strategic stability, learning outcomes are not invariant to replacing an extensive form game with the corresponding game in normal  form: In the latter case  there are no unreached information sets,  and the terminal node that is reached reveals the strategy used by each player.  

To capture what is essential for learning outcomes in the normal form, we augment it with \emph{terminal node partitions} (\citet{Fudenberg2015}, \citet{Fudenberg2018b})
 which describe the information players observe when the game is played.   We show that playing the extensive form game is equivalent to playing the normal form with the  terminal node partitions that gives players the same information as would be revealed by the terminal nodes in the extensive form, so that the two games are information invariant. We also show that  if agents play the normal form derived from an extensive form game and observe their opponents' strategies, the learning outcome is a refinement of backward induction and of  $S^{\infty}W$ (\citet{dekel1990rational}),  but does not imply iterated weak dominance.

\section{Informal Overview of the Learning Model}

 We begin with an informal overview of the learning model, deferring the full description of the learning model until Section \ref{sec:full_model}.  We consider an overlapping generations learning environment where time is discrete and doubly infinite, $t \in \{...,-2,-1,0,1,2,...\}$. There is a continuum of agents of mass $1$ in each player  role $i\in\{1,...,I\}.$ The agents have geometric lifespans, with  i.i.d. survival probability $\gamma$ per period. Each period newborn agents replace the departing agents so the sizes of the various populations are constant, and then agents are anonymously matched to play a fixed  extensive form stage game with  perfect recall.

 The game has information sets $\mathcal{H}_{i}$
 for each player $i\in\{1, ..., I\}$, with available actions $A_{h}$ at each $h\in\mathcal{H}_{i}$. A pure strategy $s_{i} \in S_{i}$
 of $i$ assigns an action $s_{i}(h)\in A_{h}$ to every information
 set $h$ of $i.$ Denote the set of terminal nodes of game tree as $Z$, and let $\mathbf{z}(s)$ denote the terminal node reached by strategy profile $s$. Player $i$ has a utility function defined on terminal nodes, $u_{i}:Z\to\mathbb{R}$ and a corresponding utility function on strategy profiles  $u_{i}(s)=u_{i}(\mathbf{z}(s))$.

 Each agent has a \emph{terminal node partition} $\mathcal{P}_i$  (\citet{Fudenberg2015}, \citet{Fudenberg2018b})
 over $Z$, and they observe which partition element contains the terminal node of their match at the end of each period.\footnote{\citet{Fudenberg2015}  analyze  settings where each player moves only once, and players who choose an \textbf{Out} action do not observe the choices made by others. The rationalizable conjectural equilibria of \citet{RCE} and \citet{esponda2013rationalizable} use \emph{signal functions} to model what players observe when the game is played. These papers do not explicitly consider extensive form games so their signal functions are more abstract.}
 In previous analyses of explicit learning models,  this partition is discrete, i.e., all agents observe the realized terminal node, and this will be our default assumption. However, in some settings it is natural to assume that agents observe less; for example, in a sealed-bid first price auction, agents might only observe the winning bid.

All agents are rational Bayesians who choose \emph{policies} (maps from history of past observations to current play) that maximize their expected discounted payoff.  They are born with priors over the prevailing steady-state distribution of play in the opponent populations, which they update using their observations.
In every period $t$, the \emph{state} of the system is the shares of agents in a given player role with the various possible histories. The state and the optimal policies  induce an   \emph{aggregate  strategy} that describes the distribution of strategies in each player-role population, and thus  an \emph{update rule} that maps states in period $t$ to states in period $t+1$.  We study this system's steady states, which are the fixed points of the update rule. 

Agent's observations can depend on their play, so their optimal policies may incorporate a value for  ``experimenting'' with various strategies that have the potential to improve  payoff.  The size of the experimentation incentive depends on their continuation probability, their discount factor $\delta \in [0,1)$, and how much they have already learned: inexperienced agents have more incentive to experiment, and they cease experimenting when they have enough data.  

We focus on the limits of steady-state play when $\gamma$ tends to $1$, so  agents can acquire enough observations to outweigh their prior. We also assume that $\delta$ goes to $1$. Otherwise, agents may not experiment enough to rule out limits that are not  Nash equilibria. We call the strategy profiles that emerge in this limit \emph{patiently stable}.

\section{Examples}

\subsection{Failures of  Forward Induction and Iterated Weak Dominance}

 We give  simple examples to show that  equilibria  that violate minimal notions of forward induction or the related  concept of iterated weak dominance can be patiently stable.
 

\subsubsection{Information Value of Dominated Strategies}\label{Payoff-Dominated Action}

Consider the following game: P1 chooses from $\textbf{Out}$, $\textbf{In1}$, and $\textbf{In2}$. If P1 chooses $\textbf{Out}$, the game is over and each player gets $0$. If P1 chooses $\textbf{In1}$ or $\textbf{In2}$, P2 plays $\textbf{L}$ or $\textbf{R}$ without knowing P1’s choice. The figure below shows the game in its extensive-form and normal-form representations.

\begin{figure}[H]
\begin{centering}
\begin{minipage}[c]{0.4\columnwidth}%
\includegraphics[scale=0.4]{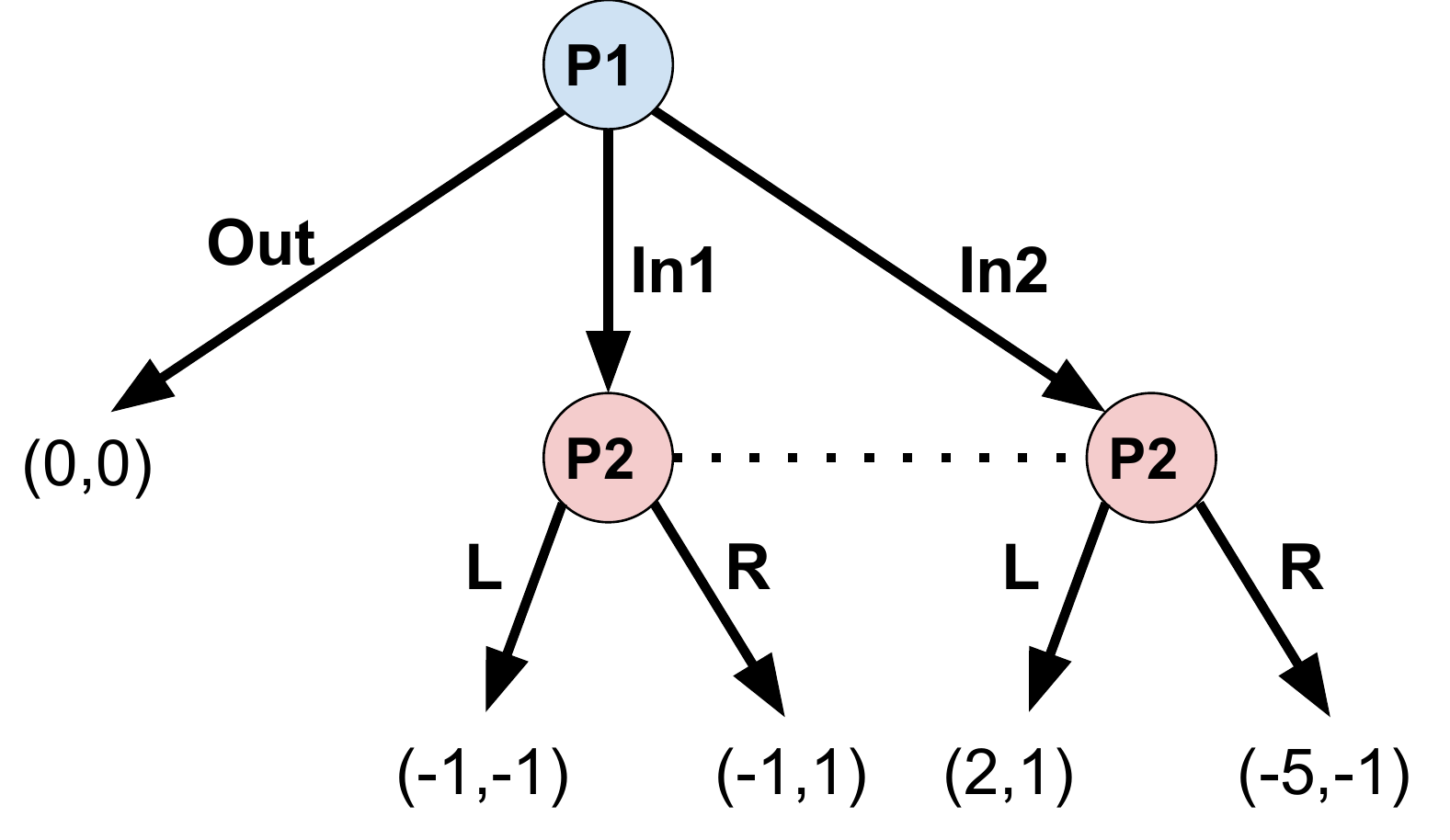}%
\end{minipage} \qquad{}%
\begin{minipage}[c]{0.4\columnwidth}%
\begin{center}
\begin{tabular}{|c|c|c|}
\hline 
 & \textbf{L} & \textbf{R}\tabularnewline
\hline 
\hline 
\textbf{Out} & 0,0 & 0,0\tabularnewline
\hline 
\textbf{In1} & -1,-1 & -1,1\tabularnewline
\hline 
\textbf{In2} & 2,1 & -5,-1\tabularnewline
\hline 
\end{tabular}
\par\end{center}%
\end{minipage}
\par\end{centering}
\caption{$\textbf{In1}$ is strictly dominated by $\textbf{Out}$ but provides the same information as \textbf{In2} and performs better than \textbf{In2} against some P2 strategies.}
\label{Example 2 Payoffs}
\end{figure}

The strategy $\textbf{In1}$ is strictly dominated by $\textbf{Out}$ for P1, and the iterated-dominance criterion of \citet{Kohlberg1986} requires that ``A solution of a game $G$ contains a solution of any game $G'$ obtained from $G$ by deletion of a dominated strategy.'' In the game $G'$ that results from the deletion of $\textbf{In1}$, $(\textbf{In2},\textbf{L})$ is the only sequential equilibrium and so the only strategically stable equilibrium. Thus the Nash equilibrium $(\textbf{Out}, \textbf{In1})$ is ruled out by forward induction. 

In contrast, when an inexperienced P1 agent plays this game and  observes the terminal node at the each of each match, the agent may find it optimal to play $\textbf{In1}$. This is because $\textbf{In1}$ and $\textbf{In2}$ are informationally equivalent experiments: they provide the same signal about how P2s play. But if P1’s current belief puts much higher probability on P2s playing $\textbf{R}$ than $\textbf{L}$, then P1’s expected payoff from $\textbf{In1}$  exceeds that of $\textbf{In2}$. A sufficiently patient P1 agent will choose to experiment and learn about P2’s play in order to figure out whether $\textbf{Out}$ or $\textbf{In2}$ is a better response against the aggregate P2  play, but the cheapest such experiment may be $\textbf{In1}$.\footnote{\citet{Fudenberg1993}, footnote 10 pointed out the possibility that this might occur in their closely related learning model but did not provide a proof that it does.} 

\begin{claim}
$(\textbf{Out}, \textbf{R})$ is a patiently stable strategy profile for the game in Figure \ref{Example 2 Payoffs}.
\label{Out is Stable with Payoff-Dominated Actions Result}
\end{claim}

In Section \ref{dominated_action_general} we establish  more general  sufficient conditions for patient stability in two-player games where each player moves at most once.  These conditions give us a class of games where patiently  stable profiles fail forward induction because of the informational value of dominated strategies. The idea of the proof is to choose ``supportive'' priors that lead the early mover to experiment in a way consistent with the desired equilibrium (such as choosing  \textbf{In1} instead of  \textbf{In2} in the example above) unless they have previously seen an out-of-equilibrium response from the second mover. 

\subsubsection{Insufficient Data to Eliminate Weakly Dominated Opponent Play}\label{sec:no_data}

In the previous example, there is a dominated strategy that is still used by a rational agent as it provides information about their opponents' aggregate play. By contrast, for the game in Figure \ref{fig:doubly_dom}, the strategy \textbf{In2} is \emph{doubly dominated} by \textbf{In1} for P2 agents: it  provides the same
information about opponent play but, whenever these actions can be played, \textbf{In2} always gives
 a strictly lower payoff than \textbf{In1}.  A rational P2 agent will therefore
never play \textbf{In2} even as an experiment, which makes it more
surprising that the learning outcome for patient and long-lived agents
can select a profile where P1 and P2 are deterred from entering by P3's \textbf{R}, which is strictly inferior to \textbf{L} against \textbf{In1}.
 
 \begin{figure}[H]
    \centering
    \includegraphics[scale=.45]{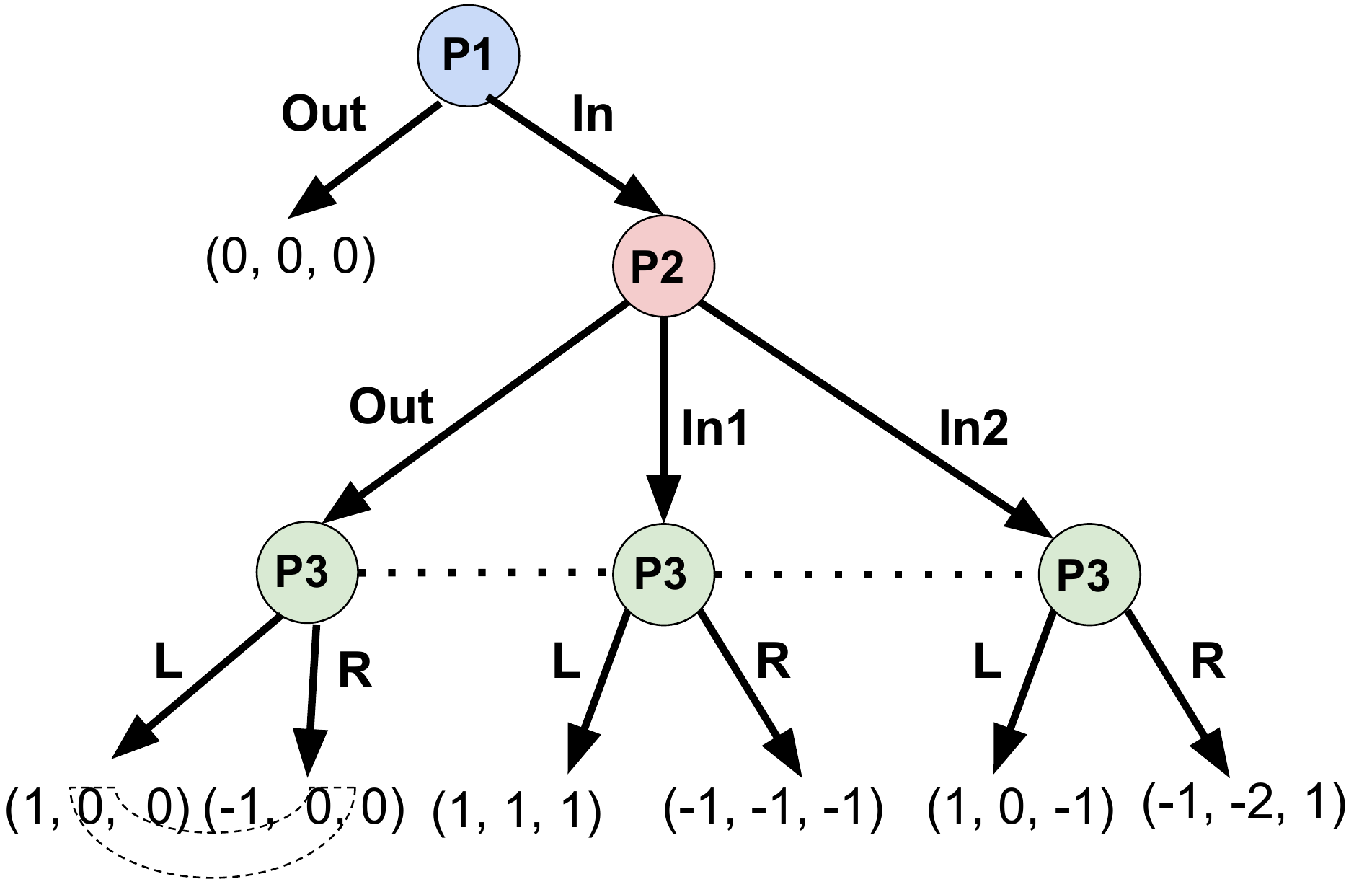}
    \caption{A three-player game where the strategy $\textbf{In2}$ is doubly dominated by $\textbf{In1}$ for P2. The dotted lines connecting two terminal nodes represent P2's terminal node partition.}
    \label{fig:doubly_dom}
\end{figure}
 Here we suppose that P1 and P3 always observe the terminal node, but P2's terminal node partition is such that they do not learn how P3 plays if they choose $\textbf{Out}$. Note that once the doubly dominated $\textbf{In2}$ is deleted for P2, $\textbf{L}$ is a strictly better strategy than  $\textbf{R}$ for P3 against any strictly mixed play of P1 and P2. But:

\begin{claim}\label{doubly_dominated_Result}
$(\textbf{Out}, \textbf{Out}, \textbf{R})$ is a patiently stable strategy profile for the game in Figure \ref{fig:doubly_dom}.
\end{claim}

The presence of a third player is critical to this conclusion. The idea is that although  aggregate P2 play puts zero probability on \textbf{In2} (as required by the elimination of weakly dominated strategies) and positive probability on \textbf{In1},  a P3 agent may not have enough data to learn this aggregate play, as P3s only observe a P2 agent entering when they encounter both a P1 and a P2 agent experimenting with some \textbf{In} action. The incentive for P2 to experiment is weak because  they are located off the equilibrium path and do not expect to play often, as in \citet{Fudenberg2006}. As a result, most P3 agents will never obtain any data to correct a prior belief that says it is more likely for P2's to choose \textbf{In2} than \textbf{In1}, so they find it optimal to play $\textbf{R}$. Even though the  aggregate steady-state play  of the P2s puts zero probability on the weakly dominated strategy,   most P3s fail to learn this. We formally analyze this example in Section \ref{sec:double_dominated_analysis}.

\subsection{Invariance}


The refinements literature following \citet{Kohlberg1986} argues that the selected set of equilibria should only depend on the  normal form, so that any two extensive forms with the same  normal form generate the same predictions. \citet{Kohlberg1986}  say that this follows from the fact that the  normal form ``captures all the relevant
information for decision purposes...''  To this we would add ``for fixed beliefs about the play of the opponents.''  Simply splitting a decision node for one player, without changing any of the information sets of the opponents, does not  change what any player observes either during the game or at the end of it, and so has no effect on the patiently stable outcomes, as in the following example.


In Figure \ref{fig:fig3}, both games have the same set of patiently stable profiles   (by Claim \ref{claim:split} below),  but the outcome $\textbf{Out}$ is only a sequential equilibrium outcome in the extensive form on the left.\footnote{In any sequential equilibrium of the game on the right,  P1 must play action 2 so P2 must play $\textbf{L}$, so P1 must play \textbf{In}.} 

\begin{figure}[H]
    \centering
    \includegraphics[scale=0.4]{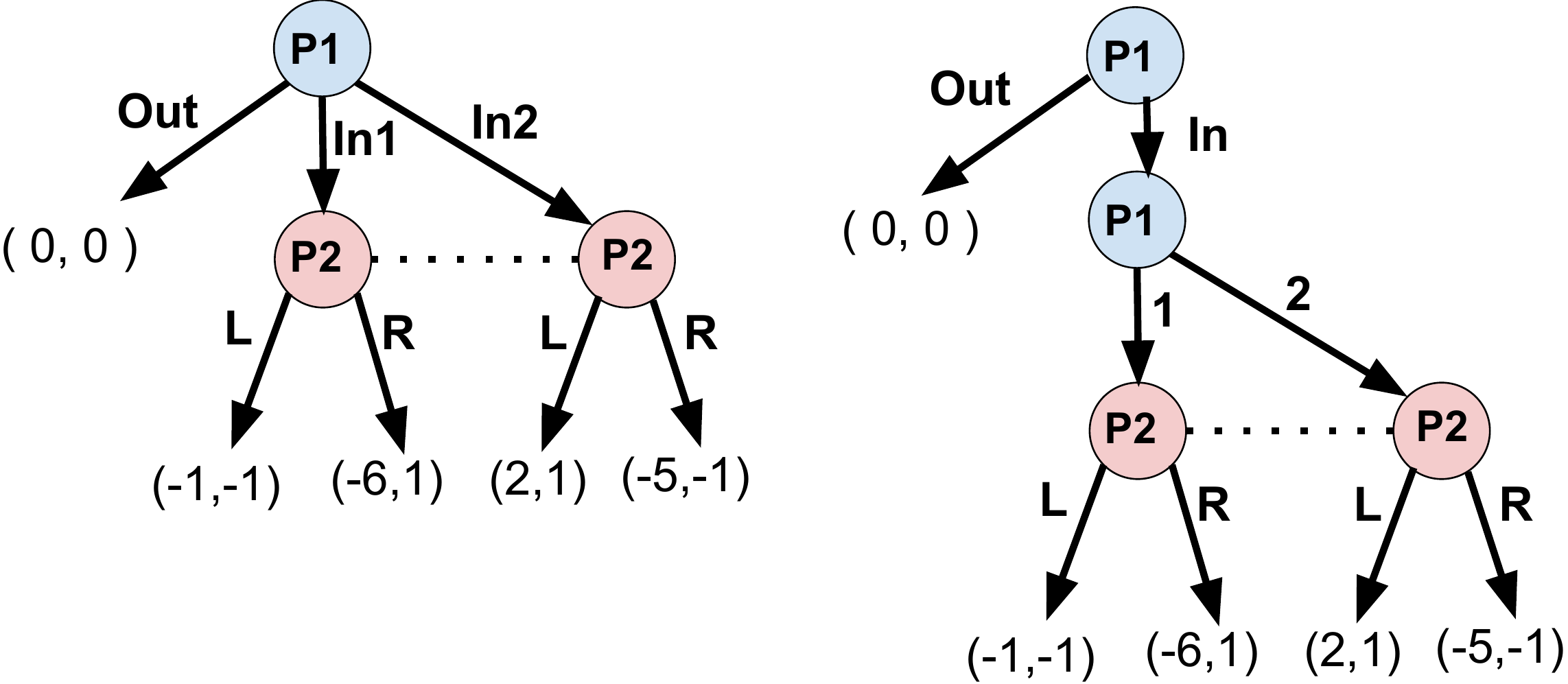}
    \caption{Both games have the same set of patiently stable profiles. But  $\textbf{Out}$ is only a sequential equilibrium outcome for the game on the left.}
    \label{fig:fig3}
\end{figure}


More generally, suppose extensive form $\hat{\mathcal{G}}$ is obtained by coalescing two consecutive information sets $h_{i}^{'}$
and $h_{i}^{''}$ of $i$ in $\mathcal{G}$ into one information set $h_{i}^{\star}$
in $\hat{\mathcal{G}}$ (according to \citet{elmes1994strategic}'s ``COA'' definition of coalescence, as in Figure \ref{fig:COA}). Then they must have the same set of patiently stable profiles. 

\begin{figure}[H]
    \centering
    \includegraphics[scale=.75]{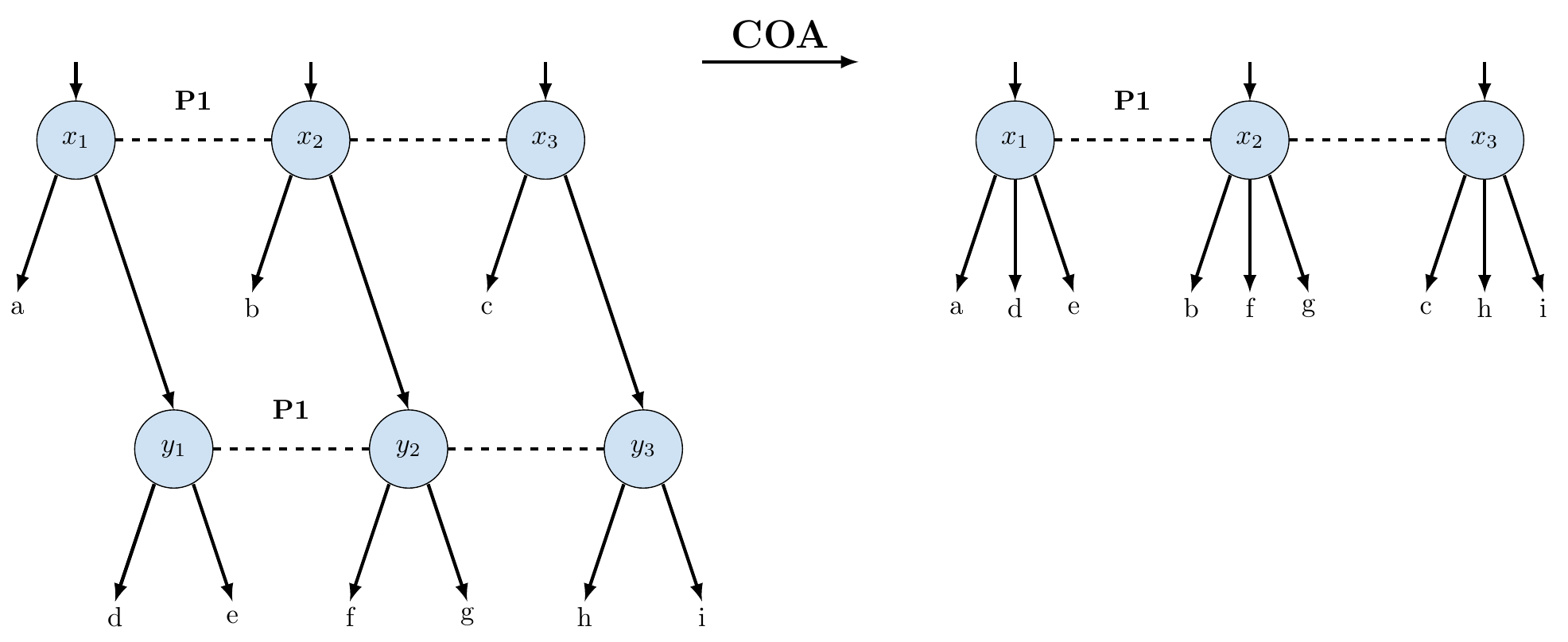}
    \caption{The coalesce operation is applied to two information sets of P1: $\{x_{1},x_{2},x_{3}\}$ and $\{y_{1},y_{2},y_{3}\}$.}
    \label{fig:COA}
\end{figure}

\begin{claim} \label{claim:split}
If $\mathcal{G}$ and $\hat{\mathcal{G}}$ are related by coalescing $h_{i}^{'}$
and $h_{i}^{''}$ into $h_{i}^{\star}$, then they have the same set of patiently stable profiles (up to identifying $i$'s two actions at $h_{i}^{'}$
and $h_{i}^{''}$ in $\mathcal{G}$ with their one action at $h_{i}^{\star}$ in $\hat{\mathcal{G}}$.)
\end{claim}

When two information sets of $i$ are coalesced, the domain of $-i$'s prior beliefs about $i$'s play must be modified so that they are about $i$'s single action at the combined information set. Appendix \ref{Proof of Claim claim:split} gives the proof of Claim \ref{claim:split}, which establishes a bijection between non-doctrinaire prior densities $g$ in the game $\mathcal{G}$  and $\hat{g}$ in the game $\hat{\mathcal{G}}$, so that the set of steady states are the same under $g$ in $\mathcal{G}$  and $\hat{g}$ in the  $\hat{\mathcal{G}}$ for any $0 \le \delta, \gamma < 1$.

However, other transformations of extensive form that leave the  normal form unchanged can change the feedback players obtain in the course of play and so change the set of patiently stable outcomes.

\begin{figure}[H]
\begin{centering}
\begin{minipage}[c]{0.4\columnwidth}%
\includegraphics[scale=1]{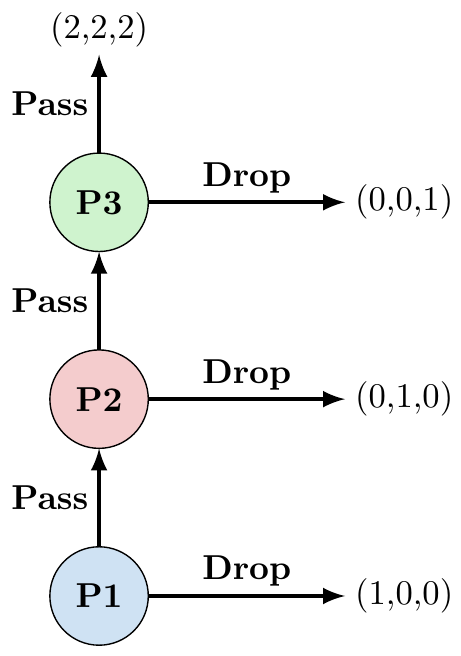}%
\end{minipage} \qquad{}%
\begin{minipage}[c]{0.4\columnwidth}%
\textbf{Pass}
\begin{center}
\begin{tabular}{|c|c|c|}
\hline 
& \textbf{Pass} & \textbf{Drop}\tabularnewline
\hline 
\hline 
\textbf{Pass} & 2,2,2 & 0,1,0\tabularnewline
\hline 
\textbf{Drop} & 1,0,0 & 1,0,0\tabularnewline
\hline 
\end{tabular}
\end{center}
\vspace*{.5em}
\textbf{Drop}
\begin{center}
\begin{tabular}{|c|c|c|}
\hline 
& \textbf{Pass} & \textbf{Drop}\tabularnewline
\hline 
\hline 
\textbf{Pass} & 0,0,1 & 0,1,0\tabularnewline
\hline 
\textbf{Drop} & 1,0,0 & 1,0,0\tabularnewline
\hline 
\end{tabular}
\par\end{center}%
\end{minipage}
\par\end{centering}
\caption{The two games have the same normal form, and (\textbf{Pass}, \textbf{Pass}, \textbf{Pass}) is the unique backward-induction profile for the extensive form on the left.  (\textbf{Drop}, \textbf{Drop}, \textbf{Pass}) is patiently stable for the game on the left, but not for the game on the right.}
\label{fig:fig4}
\end{figure}

As an example, compare the two games in Figure \ref{fig:fig4}. In the  game on the left, the unique backwards induction outcome is (\textbf{Pass}, \textbf{Pass}, \textbf{Pass}), but we know from \citet{Fudenberg2006} that the outcome (\textbf{Drop}, \textbf{Drop}, \textbf{Pass}) is also patiently stable: in the steady state, the P2s play so rarely that they choose not to experiment with  \textbf{Pass} and so never learn that the P3s  \textbf{Pass}. But this outcome is ruled out when agents play the normal form. 

\begin{claim} \label{claim:pass}
Suppose agents play the normal form on the right of Figure  \ref{fig:fig4} and observe opponents' strategies at the end of each match. Then   the only patiently stable outcome is (\textbf{Pass}, \textbf{Pass}, \textbf{Pass}). 
\end{claim}


This claims follows from Proposition \ref{prop:simple_game_BI} that we discuss later in Section \ref{sec:BI_simple_games}. In the game on the right of Figure  \ref{fig:fig4}, P3s always \textbf{Pass} because they have full-support beliefs about what others do. Unlike for the game on the left, now P2s do not need to experiment to learn this. Once P2s learn that P3s play \textbf{Pass}, they themselves play \textbf{Pass}. This means that, when agents are long-lived, the vast majority of P2s in the population play \textbf{Pass}, so P1s learn to play \textbf{Pass} over \textbf{Drop} as well. 

As these examples suggest, the problem is that the normal form does not distinguish between extensive forms that differ in what is observable during the learning process\footnote{\citet{fudenberg1993self} point out the implications of this for self-confirming equilibrium, and \citet{sorin1995note} discusses its implication for the equilibria of repeated extensive form games.}. We should only expect the set of learning outcomes to be invariant to transformations that  are both \emph{decision invariant}, i.e.,  lead to the same best responses as a function of opponent strategies, and \emph{information invariant} in the sense of providing the same feedback to the agents in their learning problems.  In the  example above, this can be done by augmenting  the normal form with the  terminal node partition shown in Figure \ref{fig:cent_dash}. 

\begin{figure}[H]
    \centering
    \includegraphics[scale=1]{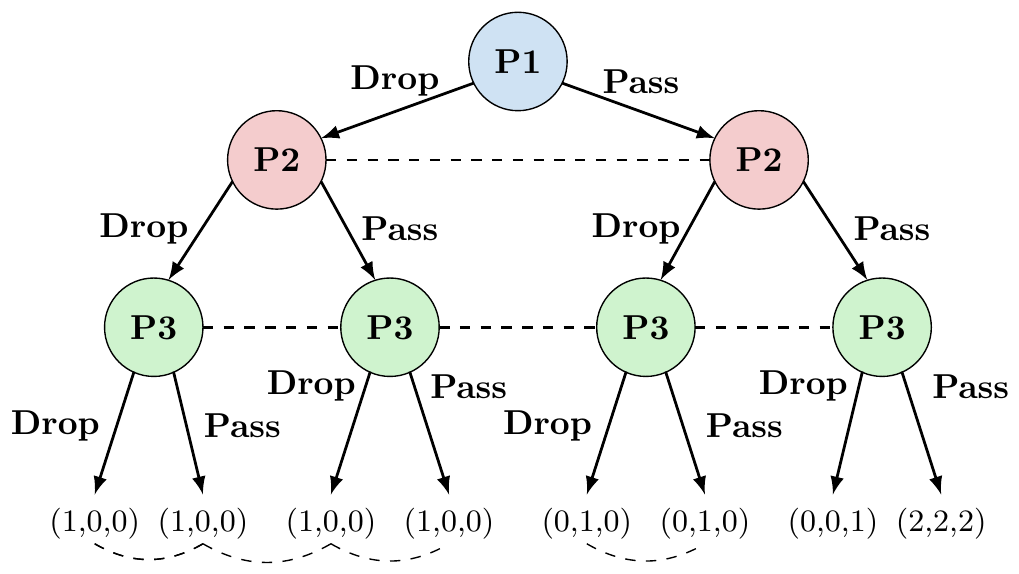}
    \caption{The game on the right of Figure \ref{fig:fig4} equipped with terminal node partitions. This game provides the same feedback to players as the game on the left of Figure \ref{fig:fig4}.}
    \label{fig:cent_dash}
\end{figure}

The partition, which is common to all three players, says that if P1 plays \textbf{Drop}, players do not observe the choices of P2 and P3, and that if P1 plays \textbf{Pass} and P2 plays \textbf{Drop} then they do not observe the choice of P3. Under this partition, (\textbf{Drop}, \textbf{Drop}, \textbf{Pass}) again becomes patiently stable.   Section \ref{sec:observability} discusses how the terminal node partitions influence which profiles are patiently stable.

\section{The Learning Model}\label{sec:full_model}


There is a unit mass population of agents who play each role $1\le i\le I$
in the game. In every period, each agent is anonymously matched with  opponents from the other populations uniformly at random to play the stage game. At the end of each play of the game, each agent observes the element of their \emph{terminal node partition} $\mathcal{P}_i$ that contains the realized terminal node of the game, where we require that $u_i(z)=u_i(z')$ if $z$ and $z'$ are  in the same cell of $i$'s terminal node partition. The agent uses this information to update their beliefs about the distribution of play in  opponent populations.

As in \citet{Fudenberg2018} and \citet{clark2021justified}, we assume that the agents have geometrically distributed lifetimes: At the end of every period, each  agent exits the system with probability $0 < 1 - \gamma \leq 1$, and a mass of newcomers is added to each population to replace the departing agents.\footnote{Previous work by  \citet{Fudenberg1993, Fudenberg2006}  assumed agents have fixed finite lifetimes. All of our results extend to this alternate lifetime specification.} Agents maximize expected discounted utility,  discounting future payoffs with a psychological discount factor  $0\le \delta < 1$.

Denote the set of pure strategies of $i$ in the game as $\mathbb{S}_i$ and the set of behavior strategies of $i$ as $\Pi_i$. Agents believe that the aggregate distribution of play in the opponent population is constant, but they do not know what that distribution is. 
Each agent in population $i$ starts with a prior belief $g_{i}\in\Delta(\times_{h\in\mathcal{H}_{-i}}\Delta(A_{h}))$ about the aggregate behavior strategy profile that describes play in opponent  populations $j \neq i$ at different information sets. We assume that, for each $i$, the prior $g_{i}$ is \emph{non-doctrinaire}, meaning that it has a density which is strictly positive on the interior of $\times_{h\in\mathcal{H}_{-i}}\Delta(A_{h})$.\footnote{The strict positivity assumption lets us appeal to the classic \citet{Diaconis1990} result on the rate of convergence of Bayesian posteriors to the empirical distribution. Note that if agents  believe that they know their opponents' payoff functions, strict positivity requires that they  assign positive probability to opponent strategies they believe are dominated. We discuss this issue more in the conclusion.}

As agents play the game and accumulate histories of past play and
observations, they update their beliefs using Bayes' rule (which is
always applicable because the priors assign positive probability to
any finite sequence of observations) and modify their behavior. Let $Y_{i,t}=(\mathbb{S}_{i}\times\mathcal{P}_{i})^{t}$
be the set of possible histories that can be observed by
an  $i$ agent of age $t$. (By convention, $\Omega^{0}=\emptyset$ for any set $\Omega$.) 
Let $Y_{i}=\cup_{t\in\mathbb{N}}Y_{i,t}$ be the
collection of all possible histories of agents from population $i$. We assume that all agents in each population $i$ use the same  optimal dynamic policy  $\mathbf{s}_{i}^{\delta,\gamma}:Y_{i}\to\mathbb{S}_{i}$
that depends on both their discount factor, $\delta$, and their lifetime
parameter $\gamma$.\footnote{ This does not mean that they all play in the same way, as  agents with the same policy may meet different opponents, and so have different histories and play different strategies.}

In every period $t$, the \emph{state} of the system, denoted $\mu_{t}=(\mu_{1, t}, ..., \mu_{I, t}) \in\times_{i}\Delta(Y_{i})$,
gives the shares of agents in the different player roles with the
various possible histories. Given $\mu_{t}$, the player $i$ policy
$\mathbf{s}_{i}^{\delta,\gamma}$ induces a player $i$ behavior
strategy $\mathbf{\sigma}_{i}^{\delta,\gamma}(\mu_{i,t})\in\Pi_{i}$
that we call the \emph{aggregate strategy} of population $i$. We call $\mathbf{\sigma}^{\delta,\gamma}(\mu_{t})=(\mathbf{\sigma}_{i}^{\delta,\gamma}(\mu_{i,t}))_{i}\in\times_{i}\Pi_{i}$
the \emph{aggregate strategy profile}.\footnote{Formally, $\mathbf{\sigma}^{\delta,\gamma}(\mu_{t})[s_{i}]=\sum_{y_{i} \in Y_i \text{ s.t. }\mathbf{s}_{i}^{\delta,\gamma}(y_{i})=s_{i}}\mu_{i,t}[y_{i}]$.}

A policy profile generates an \emph{update rule} $\mathbf{f}^{\delta,\gamma}:\times_{i}\Delta(Y_{i})\rightarrow\times_{i}\Delta(Y_{i})$,
taking the state in period $t$ to the state in period $t+1$, and
the mappings $\mathscr{R}_{i}^{\delta,\gamma}:\Pi_{-i}\rightarrow\Pi_{i}$
that describes the limit of the aggregate $i$ strategy as $t\rightarrow\infty$
when the aggregate $-i$ strategy is fixed at $\pi_{-i}$. We refer to
the mapping $\mathscr{R}^{\delta,\gamma}(\pi)\equiv(\mathscr{R}_{1}^{\delta,\gamma_{1}}(\pi_{-1}),...,\mathscr{R}_{I}^{\delta,\gamma}(\pi_{-I}))$
as the \emph{aggregate response mapping}. Similar arguments to those
in \citet{clark2021justified} show that this mapping is continuous.

We study this system's steady states, those $\mu$ satisfying $\mathbf{f}^{\delta,\gamma }(\mu)=\mu$.
We call the corresponding aggregate strategy profiles the \emph{steady
state profiles}, and denote them by $\Pi^{*}(g,\delta,\gamma)\subseteq \times _ {1 \le i \le I} \Pi_{i}$.
Again, similar arguments to those in \citet{clark2021justified} show that
these are the fixed points of the aggregate response mapping. Continuity
of the aggregate response mapping, along with Brouwer's fixed point
theorem, then implies that steady state profiles always exist.

\begin{proposition} $\Pi^{*}(g,\delta,\gamma)$ consists
of the strategy profiles that are fixed points of the aggregate response
mapping, and it is non-empty for all $g$, $\delta$,
and $\gamma$. \label{Steady State Existence Result}
\end{proposition}

When the agents are short-lived they have little chance to learn, and simply play a best response to their priors. When agents are long-lived but impatient, they do learn the steady state path of play, but need not learn how opponents respond to deviations, so any self-confirming equilibrium in strategies that are not weakly dominated could arise.  We will focus on steady states where agents are both long-lived and patient. More specifically, we focus on steady state  profiles  in the limit where agents become  long lived  ($\gamma \rightarrow 1$) and patient ($\delta \rightarrow 1$). Moreover, following the literature, we assume continuation probability $\gamma$ converges to $1$ faster than $\delta$. We call these the \emph{patiently stable} strategy profiles. The order of limits corresponds to an environment where agents are long-lived relative to their effective discount factors. This implies that people spend most of their lives myopically responding to their current beliefs.

\begin{definition}
Strategy profile $\pi$ is \textbf{\emph{patiently stable}} if there are  sequences $\{\delta_{j}\}_{j \in \mathbb{N}}$, $\{\gamma_{j,k}\}_{j,k \in \mathbb{N}}$   and associated steady-state profiles $\{\pi_{j,k} \in \Pi^{*}(g,\delta_{j},\gamma_{j,k})\}_{j,k \in \mathbb{N}}$ such that $\lim_{j \rightarrow \infty} \delta_{j} = 1$, $\lim_{k \rightarrow \infty} \gamma_{j,k} =  1$ for each $j$ and $\lim_{j \rightarrow \infty} \lim_{k \rightarrow \infty} \pi_{j,k} = \pi$.
\label{Stability Definition}
\end{definition}





The literature has previously shown that patiently stable strategy profiles must be Nash equilibria when agents observe the realized terminal nodes in the games they play.\footnote{\citet{Fudenberg1993} established this in a learning model where players had fixed finite lifetimes rather than geometric lifespans. The adaptations of these arguments given in the supplementary information of \citet{Fudenberg2018} show that this extends to geometric lifespans in general games, although the main text of \citet{Fudenberg2018} only states this result  for signaling games.} 


 Appendix \ref{sec:auxiliary_game_proof} shows that this is also true for the game and terminal node partition given in Figure \ref{fig:doubly_dom}, which is the only example in the paper that uses a  non-discrete terminal node partition to exhibit a patiently stable profile that is ruled out by classic refinements. We conjecture that patiently stable profiles must be Nash equilibria in any game provided each agent's payoff is measurable with respect to their terminal node partition, but we have not shown this. Instead, Appendix \ref{sec:example_auxiliary} gives a number of other examples from the literature where this conclusion does hold.

\section{Patient Stability,  Forward Induction, and Dominance}

\subsection{Dominated Actions in a Family of Two-Player Games}\label{dominated_action_general}

This section provides a sufficient condition for patient stability that generalizes the example from Section \ref{Payoff-Dominated Action}. We consider a family of two-player games where P1 first chooses an
action $a_{1}\in A_{1}$, which may end the game or give the play
to P2. For each P2 information set $h_{2}$, P2 chooses among the actions $\mathcal{A}_{2}(h_{2})$, and we let $\rho(h_{2})$ denote the P1 actions that lead to $h_{2}$. Write $u_{i}(a_{1},a_{2})$ for $i$'s utility at the terminal node reached by P1 playing $a_{1}$ and P2 playing $a_{2}$. We also  write $u_i(\pi_1, \pi_2)$ for $i$'s expected utility when players use behavior strategies $\pi_1$ and $\pi_2$.

We will show that equilibria of the following form are patiently stable under some non-doctrinaire prior that we construct.

\begin{compactenum}

\item P1 plays a single action $a_{1}^{*} \in A_{1}$ that uniquely maximizes their payoff given P2's strategy. (Formally, $\pi_{1}^{*}(a_{1}^{*}) = 1$ for the $a_{1}^{*}$ that satisfies $u_{1}(a_{1}^{*},\pi_{2}^{*}) > u_{1}(a_{1},\pi_{2}^{*})$ for all $a_{1} \neq a_{1}^{*}$.)

\item For each P2 information set $h_{2}$, P2 plays some response $a_{2}^{*}(h_{2})$ that is optimal given some $a_{1}^{*}(h_{2}) \in \rho(h_{2})$. Moreover, out of $\rho(h_{2})$, $a_{1}^{*}(h_{2})$ is optimal for P1 given that P2 plays $a_{2}^{*}(h_{2})$. 

\item If $a^*_1$ leads to P2 information set $h_{2}^{*}$, then $a_{2}^{*}(h_{2}^{*})$ uniquely maximizes P2's payoff against $a_{1}^{*}$.

\end{compactenum}
The $(\textbf{Out},\textbf{In2})$ equilibrium from Section \ref{Payoff-Dominated Action} is of this form: $\textbf{Out}$ serves the role of $a_{1}^{*}$, and for P2's only information set, P2's prescribed response of $\textbf{R}$ is the unique best response to $\textbf{In1}$. In turn, $\textbf{In1}$ is the best action out of $\{\textbf{In1},\textbf{In2}\}$ for P1 when P2 chooses $\textbf{R}$. 

In the equilibria we construct,  P2 may  best reply to dominated P1 actions $a_1^*(h_2)$ at some information sets $h_2$.  Nevertheless, we show in Proposition \ref{Dominated Action Stability General Result} below that every such equilibrium is patiently stable, which implies Claim \ref{Out is Stable with Payoff-Dominated Actions Result}. 



\begin{proposition}
Suppose that $\pi^{*}$ is an equilibrium of the form given above. Then $\pi^{*}$ is patiently stable for any pair of non-doctrinaire P1 and P2 priors that are supportive of $\pi^{*}$.
\label{Dominated Action Stability General Result}
\end{proposition}

The key is to choose priors that are ``supportive'' of the equilibrium. A supportive P1 prior is such that, for every off-path P2 information set $h_{2}$, a P1 agent prefers to experiment with $a_{1}^{*}(h_{2})$ over any other action in $\rho(h_{2})$ unless they have previously experienced a P2 response at $h_{2}$ for which $a_{1}^{*}(h_{2})$ is not conditionally optimal. Similarly, a supportive P2 prior leads P2 agents to want to play $a_{2}^{*}(h_{2})$ at an information set $h_{2}$ unless they have previously witnessed a P1 agent play some action in $\rho(h_{2})$ other than $a_{1}^{*}(h_{2})$. These properties are formalized in Appendix \ref{Proof of Proposition Dominated Action Stability General Result}, which also contains the proof of Proposition \ref{Dominated Action Stability General Result}.

\subsection{Stability and Doubly Dominated Actions: An Example}\label{sec:double_dominated_analysis}

The example from Section \ref{sec:no_data} does not fit with the sufficient conditions for stability we gave in Section \ref{dominated_action_general}: it involves P3 best replying to the action  $\textbf{In2}$ by P2, a doubly dominated action that is not optimal among the P2 actions $\{ $\textbf{In1}$,  $\textbf{In2}$ \}$ that reach the same P3 information set. We use a different argument to show that the $(\textbf{Out}, \textbf{Out}, \textbf{R})$ outcome is patiently stable.

\begin{proposition}
For the game in Figure \ref{fig:doubly_dom}, (\textbf{Out}, \textbf{Out}, \textbf{R}) is a patiently stable profile for any non-doctrinaire P1 prior $g_{1}$, non-doctrinaire P2 prior $g_{2}$ under which the expected probability of $\textbf{L}$ is strictly less than $1/2$, and non-doctrinaire P3 prior $g_{3}$ that leads a P3 agent to only play $\textbf{L}$ when they have previously observed a P2 agent play $\textbf{In1}$.
\label{Relatively Experienced P1 Stability Result}
\end{proposition}

This proposition specifies the prior beliefs that make patient stability hold in Claim \ref{doubly_dominated_Result}. The proof of this result in Appendix \ref{Proof of Proposition Relatively Experienced P1 Stability Result} first notes that P1 observes P3's play if and only if they experiment with \textbf{In}. This lets us  bound the number of periods that P1s will typically experiment with  \textbf{In} before becoming pessimistic and switching to  \textbf{Out} forever in a steady state where P3s play \textbf{R} with high enough probability, so most  P2 agents will learn that their information set is rarely reached.  Thus they will   choose  $\textbf{Out}$ instead of experimenting with $\textbf{In1}$,  since they do not value information they will rarely get to use. This lets us construct a steady state where most P3 agents have never observed any instance of matched P2 agents choosing any action other than $\textbf{Out}$, and therefore choose \textbf{R} based on their prior belief.

\section{Observability and Patiently Stable Profiles}\label{sec:observability}

In this section, we study the effect of what agents observe at the end of each play of the game  on the patiently stable profiles.  Sections \ref{sec:BI_simple_games} and \ref{sec:iterative} show that in normal forms arising from simple games, where agents always observe matched opponents' extensive-form strategies, patiently stable profiles must select the same outcome as the backward induction outcome of the original game. Section \ref{sec:native_partition} says that if the normal form of an extensive form is equipped with the right terminal node partitions, it leads to the same patiently stable profiles as the extensive form. Section \ref{sec:coarser} provides an example where patiently stable profiles satisfy the iterated deletion of weakly dominated strategies with coarser observations but not finer ones.

\subsection{Backward Induction in Simple Games when Agents Observe Strategies}\label{sec:BI_simple_games}
A \emph{simple game} is an extensive-form game of perfect information where no one moves more than once along any path and no player is indifferent between any two terminal nodes, so  there exists a unique backward induction strategy profile. 

Consider the normal form of the simple game where agents simultaneously choose strategies from the extensive-form game tree and observe  opponents' strategies at the end of the match.  The next result shows that the only patiently stable profile of the normal form is the backward induction strategy profile.   In fact, we show something stronger: this is the only profile that is \emph{$\delta$-stable}.

\begin{definition}
For $0 \le \delta < 1$ and non-doctrinaire prior $g$, strategy profile $\pi$ is \textbf{\emph{$\mathbf{\delta}$-{stable}}} under $g$ 
if there is a collection of parameter sequences $\{\gamma_{k}\}_{k \in \mathbb{N}}$  and associated steady-state profiles $\{\pi_{k} \in \Pi^{*}(g,\delta,\gamma_{k})\}_{k \in \mathbb{N}}$ such that $\lim_{k \rightarrow \infty} \gamma_{k} =  1$  and $\lim_{k \rightarrow \infty} \pi_{k} = \pi$.
\label{delta stable definition}
\end{definition}

\begin{proposition}
\label{prop:simple_game_BI} Suppose agents play the normal-form
representation of a simple game. Then, every $\delta$-stable profile puts probability $1$ on a
backward-induction outcome. 
\end{proposition}


In particular, this implies that for the normal form in Figure \ref{fig:fig4}, (\textbf{Pass}, \textbf{Pass}, \textbf{Pass}) is the only learning outcome when agents are sufficiently long lived. As we show in Appendix \ref{Proof of Proposition prop:simple_game_BI}, Proposition \ref{prop:simple_game_BI} follows from a more general result in the next section about patient stability in  normal form games. 

\subsection{An Iterative Deletion Refinement in Normal Forms} \label{sec:iterative}

The next proposition discusses the implications of patient stability in environments of ``maximal observability'': that is, agents play the normal form derived from an extensive form game and observe their opponents' strategies. This result gives us a benchmark of what long-lived agents will learn in games if they do not need to experiment. The result takes the form of an iterative procedure that eliminates at
each step some of the remaining strategies that do not best respond
to strictly mixed conjectures that put arbitrarily low conditional probabilities
on eliminated opponent strategies. Let $\mathcal{S} = \{\times_{i} \widetilde{S}_{i} : \forall i, ~ \widetilde{S}_{i} \subseteq S_{i}\}$ be the set of product spaces generated by the subsets of the player strategy spaces. 

\begin{definition} A sequence $(S^{(0)}, D^{(0)}),(S^{(1)},D^{(1)})... \in \mathcal{S}^{2}$ is a \textbf{\emph{valid elimination sequence}} if
\begin{compactenum}

\item For each $i$, 
$S_{i}^{(0)}=S_{i} \setminus D_{i}^{(0)},$ and $D_{i}^{(0)}$ is any subset of $i$'s weakly dominated strategies,

\item  For each $i$ and $m > 0$, $D_{i}^{(m)}$ is a subset of $S_{i}^{(m-1)}$
such that, for every $s_{i} \in D_{i}^{(m)}$, there exists some $\epsilon>0$ where  $\mathbb{E}_{\sigma_{-i}}[u_{i}(s_{i},s_{-i})] < \max_{s_{i}' \in S_{i}} \mathbb{E}_{\sigma_{-i}}[u_{i}(s_{i}',s_{-i})]$ for all correlated opponent strategy profiles $\sigma_{-i} \in \Delta(S_{-i})$ satisfying \textup{$\sigma_{-i}(S_{j}^{(m-1)}|s_{-ij})\ge1-\epsilon$}
for every $j\ne i$ and $s_{-ij} \in S_{-ij}$, and

\item For each $i$ and $m > 0$, $S_{i}^{(m)}=S_{i}^{(m-1)}\backslash D_{i}^{(m)}.$

\end{compactenum}
\label{Valid Sequence Definition}
\end{definition}

In a valid elimination sequence, at every stage of the iteration, the only player $i$ strategies that can be eliminated are those for which the following condition holds: There is an  $\epsilon > 0$ such that the strategy is suboptimal under any conjecture that, for each opponent $j$, puts probability at least $1 - \epsilon$ on  $j$ strategies that have not yet been eliminated conditional on any strategy profile of the opponents other than $j$.

\begin{proposition}

\label{prop:normal_form_elimination} 
For a valid elimination sequence $(S^{(0)}, D^{(0)}),(S^{(1)},D^{(1)})... \in \mathcal{S}^{2}$, let $S_{i}^{*}=\cap_{m=0}^{\infty}S_{i}^{(m)}$. If agents observe matched opponents'
strategy choices at the end of each game, then every $\delta$-stable
strategy profile is supported on the non-empty set $\times_{i}S_{i}^{*}$. 
\end{proposition}

The idea behind the proof is that agents never use weakly dominated strategies in $D_i^{(0)}$ because they have full-support  beliefs about others' play, and experienced agents learn that these strategies are rarely used by an extension of \citet{Diaconis1990}'s result in \citet{pathwise}. This implies strategies in $D_i^{(1)}$ only get used with very low probabilities in the steady state, as they are only played  by the very young agents. Iterating this argument lets us eliminate the strategies in $D_i^{(2)}, D_i^{(3)}$, and so forth.

Different valid elimination sequences may lead
to different strategy sets $S_{i}^{*}$ in the end. Proposition \ref{prop:normal_form_elimination}, which we prove in Appendix \ref{Proof of Proposition prop:normal_form_elimination} gives a family of necessary conditions of patient stability, corresponding
to different valid elimination sequences. 

Some of the valid  elimination sequences
 correspond to well-known solution concepts. One example is backward induction in simple games: Proposition \ref{prop:simple_game_BI} follows from Proposition \ref{prop:normal_form_elimination} by letting $D_{i}^{(m)}$ be
those extensive-form strategies of $i$ that are inconsistent with backward induction at some decision node $m+1$ steps away from the terminal nodes, but agree with it at all decision nodes $m$ or fewer steps away from the terminal nodes.  The proof of Proposition \ref{prop:simple_game_BI} verifies that these $D_{i}^{(m)}$ form a valid elimination sequence.

A second example is the solution concept $S^{\infty}W$ (\citet{dekel1990rational}), which \citet{borgers1994weak} shows is equivalent to players having full support beliefs about the play of others and that this and the rationality of the players are ``almost common knowledge.'' This solution concept corresponds
to choosing $D_{i}^{(0)}$ to be all weakly dominated strategies of
$i$ in the original game, and, at each step $m,$ choosing $D_{i}^{(m+1)}\subseteq S_{i}^{(m)}$
to be the strictly dominated strategies of $i$ in the reduced game
where $i$ has the strategy set $S_{i}^{(m)}.$  To see that this is a valid elimination sequence, note that if $s_{i}$
is strictly dominated, then there is some $\sigma_{i}\in\Delta(S_{i}^{(m)})$
and $\eta>0$ so that $u_{i}(\sigma_{i},s_{-i})>u_{i}(s_{i},s_{-i})+\eta$
for all $s_{-i}\in S_{-i}^{(m)}.$ By continuity, there exists some
$\epsilon>0$ so that for any full-support correlated opponent strategy $\sigma_{-i}$
of the original game where $\sigma_{-i}(S_{-i}^{(m)})\ge1-\epsilon$, we have $u_{i}(\sigma_{i},\sigma_{-i})>u_{i}(s_{i},\sigma_{-i})+\eta/2$,
so in particular $s_{i}$ is not a best response to any such $\sigma_{-i}$.

While the refinement in Proposition \ref{prop:normal_form_elimination} is stronger than $S^\infty W$, it is weaker than iterated elimination of weakly dominated strategies.  This is because in defining $D_i^{(m)}$ in the iterative procedure, we consider conjectures where the probabilities assigned to deleted strategies can be arbitrarily small, but need not be zero. Provided there are at least two remaining strategies, this does not imply that the highest probability assigned to a deleted strategy must be lower than the lowest probability assigned to a remaining strategy. This distinguishes the Proposition \ref{prop:normal_form_elimination}  refinement from other refinement concepts like the iterated admissibility of  \citet{brandenburger2008admissibility} and the consistent pairs of \citet{borgers1992cautious}.\footnote{Consistent pairs capture the implications of assuming that players  maximize expected utility, and that players form ``cautious expectations.'' Such pairs are only defined for two-player games, and do not always exist.} For instance, for the game in Figure \ref{fig:IA}, there is no valid elimination sequence that uniquely selects the (\textbf{A}, \textbf{X}) strategy profile, even though (\textbf{A}, \textbf{X}) is the unique iteratively admissible profile. From a learning perspective, the idea is that although   \textbf{C} is strictly dominated for P1, if P1 always play \textbf{B} then P2 can still maintain a belief that  \textbf{C} is relatively more likely than \textbf{A} and thus choose \textbf{Y}. Indeed, it is easy to see that (\textbf{B}, \textbf{Y}) is a steady-state profile for any $0\le\delta,\gamma<1$ (and therefore, patiently stable) if P1 starts with a strong prior belief that P2s play \textbf{Y} and P2s start with a Dirichlet prior with weights (1, 1, 10) on the P1 actions (\textbf{A}, \textbf{B}, \textbf{C}).

\begin{figure}
\begin{centering}
\begin{tabular}{|c|c|c|}
\hline 
 & \textbf{X} & \textbf{Y}\tabularnewline
\hline 
\hline 
\textbf{A} & 2, 2 & 0, 0\tabularnewline
\hline 
\textbf{B} & 1, 1 & 1, 1\tabularnewline
\hline 
\textbf{C} & -10, 0 & -10, 1\tabularnewline
\hline 
\end{tabular}
\par\end{centering}
\caption{In this game, (\textbf{A}, \textbf{X}) is the unique iteratively admissible
outcome \citep{brandenburger2008admissibility}, but (\textbf{B}, \textbf{Y}) is also patiently stable.} \label{fig:IA}
\end{figure}

\subsection{Information-Equivalent Normal Forms}\label{sec:native_partition}
For an extensive form $\mathcal{G}$ and terminal node partitions $\mathcal{P}$, consider the normal form $\mathcal{N}$
whose terminal nodes correspond to strategy profiles in $\mathcal{G},$
that is $\times_{i}\mathbb{S}_{i}$. Learning from the terminal node partition $\mathcal{P}$ in  $\mathcal{G}$ and learning in $\mathcal{N}$ (with the standard assumption that the normal-form strategies played are observed by all the players at the end of each game) lead to different patiently stable
profiles in general, as shown above. However, when $\mathcal{N}$ is equipped with the appropriate terminal
node partitions, it will have the same set of patiently stable profiles
as $\mathcal{G}$.

The $\mathcal{P}-$\emph{equivalent} terminal node partitions are 
$\mathcal{\hat{P}}_{i}$ for $i$ in $\mathcal{N}$ are such that $\mathcal{\hat{P}}_{i}(s)=\mathcal{\hat{P}}_{i}(s')$
if and only if $\mathcal{P}_{i}(\mathbf{z}(s))=\mathcal{P}_{i}(\mathbf{z}(s'))$. Players hold beliefs over opponents' behavior strategies in $\mathcal{G}$
and mixed strategies in $\mathcal{N}$, but we can transform a non-doctrinaire
belief over behavior strategies into one over mixed strategies and
vice versa when $\mathcal{G}$ has perfect recall, by Kuhn's theorem.

\begin{proposition} \label{prop:naive_terminal_nodes}
The patiently stable profiles
of $(\mathcal{G},\mathcal{P})$ are the same as the patiently stable profiles of $\mathcal{N}$
with the $\mathcal{P}-$equivalent terminal node partitions.
\end{proposition}

Intuitively,  the  definition of  $\mathcal{\hat{P}}_{i}$ implies agents have the same feedback in the two games,  so the problems are information invariant, and the normal form and extensive form are decision invariant. We formally show this in Appendix \ref{Proof of Proposition prop:naive_terminal_nodes}.\footnote{Note that unlike the ``normal form information sets'' of \citet{mailath1991extensive}, the equivalent terminal node partition cannot be derived from the normal form alone.}

\subsection{Coarser Terminal Partitions May Eliminate Patiently Stable Profiles}\label{sec:coarser}

Sections \ref{sec:BI_simple_games} and \ref{sec:iterative} show that coarser observations of opponents' strategies can expand the set of patiently stable profiles. 
But this is not always true, and coarser terminal node partitions can shrink rather than expand the set of patiently stable profiles in other games. 

\begin{figure}
    \centering
    \includegraphics[scale=.3]{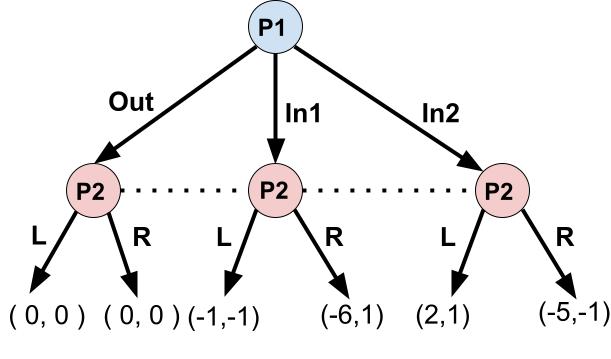}
    \includegraphics[scale=.3]{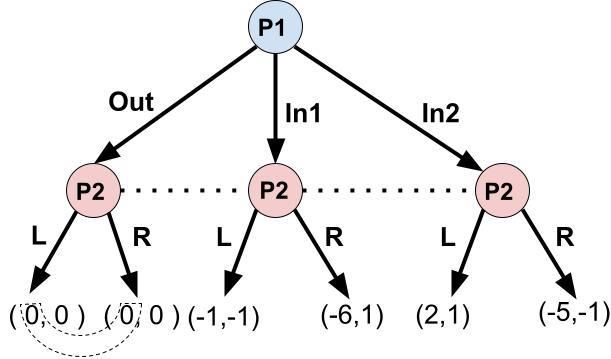}
    \caption{In the game on the left, both players observe the terminal node. In the game on the right, P1 does not observe P2's play if they choose \textbf{Out}. The (\textbf{Out}, \textbf{R}) profile is patiently stable for the game on the left, but not for the game on the right.}
    \label{fig:comparing_partitions}
\end{figure}

Consider the two games in Figure \ref{fig:comparing_partitions} that only differ in the terminal node partition of P1. In the game on the left with the finer terminal node partitions, (\textbf{Out}, \textbf{R}) is patiently stable. It is easy to see that if P1 and P2 both start with a strong prior belief in the (\textbf{Out}, \textbf{R}) equilibrium and P2 thinks \textbf{In1} is more likely than \textbf{In2} when they have only seen P1s play \textbf{Out}, then it is a steady state under any $0\le \delta, \gamma < 1$  for  (\textbf{Out}, \textbf{R}) to be played in every match. 

But,   (\textbf{Out}, \textbf{R}) is not patiently stable in the game on the right with the coarser terminal node partitions, as we show in Claim \ref{claim:coarser_partition}.\footnote{Technically, Claim \ref{claim:coarser_partition} imposes additional restrictions on the P2 prior, but the class of priors allowed is broad and includes those with densities that are strictly positive and continuous everywhere as well as Dirichlet distributions.} The proof idea, which we rigorously demonstrate in Appendix \ref{Proof of Claim claim:coarser_partition}, is that patient P1 players will spend many periods experimenting with \textbf{In2}, since they cannot learn P2's play by choosing \textbf{Out}. This teaches P2s that P1s are much more likely to use \textbf{In2} than \textbf{In1}, so that they should not play \textbf{R}. 

\begin{claim}
\label{claim:coarser_partition}
For the game on the right in Figure \ref{fig:comparing_partitions}, suppose P2's prior belief satisfies Condition $\mathcal{P}$
from \citet{fudenberg2017bayesian}. Then, every patiently stable profile satisfies $\pi(R)=0.$
\end{claim}

Note that \textbf{In1} is strictly dominated for P1, and if P2 thinks that P1 never plays \textbf{In1}, then \textbf{L} is strictly better than \textbf{R} for P2 given any conjecture that puts positive probabilities on both \textbf{Out} and \textbf{In2}. Thus (\textbf{Out}, \textbf{R}) is ruled out by iterated elimination of weakly dominated strategies, and stable learning outcomes in the example violate this refinement with a finer terminal node partition, but not with a coarser one. 

\section{Conclusion}
The implications of learning depend crucially on the structure of the game and on what agents observe about others’ play. When the game and the feedback structure make it profitable for patient players to experiment with dominated strategies (for instance, when agents get no information from choosing a safe action but can use a worse safe action to learn about the consequences of a risky action), patiently stable profiles may violate forward induction or iterated weak dominance. When agents must experiment to learn about off-path play, patiently stable profiles may violate backward induction. But if agents  observe opponents’ strategies regardless of their own play, patiently stable profiles always satisfy backward induction in simple games. This shows that ruling out some Nash equilibria requires close attention to the details of the game and the learning environment. 

As in previous work, we have assumed that agents have non-doctrinaire priors in order to appeal to the \citet{Diaconis1990} result on the speed of convergence of Bayesian posteriors to the empirical distribution. \citet{pathwise} extends their convergence result to priors without full support, but if the true state is outside of the support of the priors then agents need not stop experimenting in finite time, as shown by \citet{fudenberg2017active}. This  raises a suite of new issues, as  patiently stable states might not be  Nash equilibria. 

Also, the assumption that agents have non-doctrinaire priors over aggregate play in the other populations rules out settings where agents place probability $0$ on opponent strategies that they believe are strictly dominated. Since much of the refinements literature implicitly assumes all players know the payoff functions of the others,  it is natural to wonder if adding some forms of restrictions on the priors would bring the patiently stable outcomes closer to the predictions of classic equilibrium refinements. In the case of signaling games with independent priors, \citet{fudenberg2020payoff} shows that the answer is ``yes,'' but  the implications of payoff information in general games are unclear. One issue is that, as we have seen, agents may choose to use dominated strategies for their information value, and an agent whose prior gave these strategies probability $0$ would be unable to form a Bayesian posterior.\footnote{This problem does not arise in signaling games with independent  priors, as there the senders would never experiment with dominated strategies, and receivers never experiment at all.} 
Of course, this problem does not arise with myopic players, for they will never pay a current cost to obtain information. But with myopic players there is no reason to expect learning to lead to a Nash equilibrium, let alone a refinement of it. 


\newpage

\setstretch{1}
\printbibliography

\setstretch{1.5}
\newpage
\appendix
\begin{center}
{\LARGE{}Appendix}{\LARGE\par}
\par\end{center}
\section{Omitted Proofs}
\label{Appendix}

\subsection{Proof of Claim \ref{claim:split}}
\label{Proof of Claim claim:split}

Suppose we have
$A_{h_{i}^{'}}=\{a_{1},...,a_{m},a_{\text{pass}}\}$ with the action
$a_{\text{pass}}$ leading to $h_{i}''$, $A_{h_{i}''}=\{a_{m+1},...,a_{m+n}\}$,
and $A_{h_{i}^{\star}}=\{a_{1},...,a_{m},a_{m+1},...,a_{m+n}\}.$ (For example, in Figure \ref{fig:COA}, $m=1$ and $n=2$.) Let $\Delta^{\circ}(X)$ the distributions on $X$ that
assign strictly positive probability to each very element in $X.$ We
define $\phi:\Delta^{\circ}(A_{h_{i}^{'}})\times\Delta^{\circ}(A_{h_{i}''})\to\Delta^{\circ}(A_{h_{i}^{\star}})$,
such that $\phi(\alpha_{h_{i}^{'}},\alpha_{h_{i}^{''}})(a_{k})=\alpha_{h_{i}^{'}}(a_{k})$
for $1\le k\le m$, while $\phi(\alpha_{h_{i}^{'}},\alpha_{h_{i}^{''}})(a_{k})=\alpha_{h_{i}^{'}}(a_{\text{pass}})\cdot\alpha_{h_{i}^{''}}(a_{k})$
for $m+1\le k\le m+n.$ That is, $\phi(\alpha_{h_{i}^{'}},\alpha_{h_{i}^{''}})$
is a way to choose an element of $A_{h_{i}^{\star}}$ by using $\alpha_{h_{i}^{'}}$
and $\alpha_{h_{i}^{''}}$ sequentially: first draw an element from
$A_{h_{i}^{'}}$ according to $\alpha_{h_{i}^{'}}$ and then, if the
chosen element is $a_{\text{pass}},$ draw an element from $A_{h_{i}^{''}}$
according to $\alpha_{h_{i}^{''}}.$ The map $\phi$ is one-to-one,
because $\phi(\alpha_{h_{i}^{'}},\alpha_{h_{i}^{''}})$ and $\phi(\beta_{h_{i}^{'}},\beta_{h_{i}^{''}})$
generate different distributions on $\{a_{1},...,a_{m})$ if $\alpha_{h_{i}^{'}}\ne\beta_{h_{i}^{'}}$,
while $\phi(\alpha_{h_{i}^{'}},\alpha_{h_{i}^{''}})$ and $\phi(\alpha_{h_{i}^{'}},\beta_{h_{i}^{''}})$
generate different distributions on $\{a_{m+1},...,a_{m+n}\}$ if
$\alpha_{h_{i}^{''}}\ne\beta_{h_{i}^{''}}$ and $\alpha_{h_{i}^{'}}$
a positive probability to $a_{\text{pass}}.$ Also, $\phi$ is onto,
because for a given $\alpha_{h_{i}^{\star}}\in\Delta^{\circ}(A_{h_{i}^{\star}}),$
let $\alpha_{h_{i}'}\in\Delta^{\circ}(A_{h_{i}^{'}})$ be such that
$\alpha_{h_{i}'}(a_{k})=\alpha_{h_{i}^{\star}}(a_{k})$ for $1\le k\le m$,
$\alpha_{h_{i}'}(a_{\text{pass}})=1-\sum_{k=1}^{m}\alpha_{h_{i}^{\star}}(a_{k}),$
and $\alpha_{h_{i}''}(a_{k})=\frac{\alpha_{h_{i}^{\star}}(a_{k})}{\sum_{j=m+1}^{m+n}\alpha_{h_{i}^{\star}}(a_{j})}$
for $m+1\le k\le m+n.$ It is clear that by construction, $\phi(\alpha_{h_{i}'},\alpha_{h_{i}''})=\alpha_{h_{i}^{\star}}.$
We have $\phi(\alpha_{h_{i}^{'}},\alpha_{h_{i}^{''}})=\alpha_{h_{i}^{\star}}$
if and only if $(\alpha_{h_{i}^{'}},\alpha_{h_{i}^{''}})$ and $\alpha_{h_{i}^{\star}}$
generate the same choice probabilities over the final actions $\{a_{1},...,a_{m},a_{m+1},...,a_{m+n}\}.$

For each agent $j$ in game $\mathcal{G}$, let $g_{j}:\times_{h\in\mathcal{H}_{-j}}\Delta(A_{h})\to\mathbb{R}_{+}$
be $j$'s prior prior density over $-j$'s 
strategies. Let $\hat{\mathcal{H}}_{i}$ represent $i$'s information
sets in $\hat{\mathcal{G}}$, and continue to use $\mathcal{H}_{j}$ for
$j$'s information sets in $\hat{\mathcal{G}}$ for agents $j\ne i.$ Let
$\hat{g}_{j}:\times_{h\in\mathcal{\hat{H}}_{-j}}\Delta(A_{h})\to\mathbb{R}_{+}$
be a density of $j$'s belief about $-j$'s play in $\hat{\mathcal{G}}$ such that (1) if $j\ne i,$
then $\hat{g}_{j}(\alpha_{h_{i}^{\star}},(\alpha_{h})_{h\in\mathcal{\hat{H}}_{-j}\backslash\{h_{i}^{\star}\}})=g_{j}(\phi^{-1}(\alpha_{h_{i}^{\star}}),(\alpha_{h})_{h\in\mathcal{H}_{-j}\backslash\{h_{i}',h_{i}''\}})/(\phi^{-1}(\alpha_{h_{i}^{\star}})(a_{\text{pass}}))^{n-1}$
for all strictly mixed actions $\alpha_{h_{i}^{\star}},(\alpha_{h})_{h\in\mathcal{\hat{H}}_{-j}\backslash\{h_{i}^{\star}\}}$;
(2) for $i,$ $\hat{g}_{i}=g_{i}$. That is, $\hat{g}_{j}$ is over
a different domain than $g_{j}$ since $i$ has one fewer information
set in $\hat{\mathcal{G}}$ than $\mathcal{G},$ but we identify each strictly
mixed $\alpha_{h_{i}^{\star}}$ in the domain of $\hat{g}_{j}$ with
$\phi^{-1}(\alpha_{h_{i}^{\star}})$ in the domain of $g_{j}$ and re-normalize appropriately. Note $g_{j}$ is strictly positive
on the interior and $0<\phi^{-1}(\alpha_{h_{i}^{\star}})(a_{\text{pass}})<\infty,$
so $\hat{g}_{j}$ is also strictly positive on the interior. This shows the constructed
prior $\hat{g}$ is non-doctrinaire.

By
the definition of $\phi$, each action in $\{a_{1},...,a_{m},a_{m+1},...,a_{m+n}\}$
has the same likelihood under $\alpha_{h_{i}^{\star}}$ and $\phi^{-1}(\alpha_{h_{i}^{\star}})$. Also, for every open set $E\subseteq\Delta^{\circ}(A_{h_{i}^{'}})\times\Delta^{\circ}(A_{h_{i}''})$,
the probability that $g_{j}$ assigns to $E$ is the same as the probability
that $\hat{g}_{j}$ assigns to $\phi(E)\subseteq\Delta^{\circ}(A_{h_{i}^{\star}}).$
 Note that for each $\alpha_{h_{i}^{'}}\in\Delta^{\circ}(A_{h_{i}^{'}})$,
the projection $E_{\alpha_{h_{i}^{'}}}:=\{\alpha_{h_{i}^{''}}\in\Delta^{\circ}(A_{h_{i}''}):(\alpha_{h_{i}^{'}},\alpha_{h_{i}^{''}})\in E\}$
can be viewed as a subset of $\Delta_{n}^{1}:=\{x_{m+1},...,x_{m+n-1}\ge0\text{ s.t. }x_{m+1}+...+x_{m+n-1}\le1\}\subseteq\mathbb{R}^{n-1}$.
On the other hand, the image of this projection $\phi(\{\alpha_{h_{i}^{'}}\}\times E_{\alpha_{h_{i}^{'}}})$
can be viewed as a subset of $\Delta_{n}^{\alpha_{h_{i}^{'}}(a_{\text{pass}})}:=\{x_{m+1},...,x_{m+n-1}\ge0\text{ s.t. }x_{m+1}+...+x_{m+n-1}\le\alpha_{h_{i}^{'}}(a_{\text{pass}})\}\subseteq\mathbb{R}^{n-1}$.
Both $\Delta_{n}^{1}$ and $\Delta_{n}^{\alpha_{h_{i}^{'}}(a_{\text{pass}})}$
are $n-1$ dimensional polytopes, and the latter's volume is $(\alpha_{h_{i}^{'}}(a_{\text{pass}}))^{n-1}$
that of the former. The normalizing factor $1/(\phi^{-1}(\alpha_{h_{i}^{\star}})(a_{\text{pass}}))^{n-1}$
ensures $g_j(E) = \hat{g}_j(\phi(E))$,

Combining the two observations in the previous paragraph, we see that
no matter which terminal node is observed, the posterior of $g_{j}$
will again assign the same probability to $E$ as the posterior of
$\hat{g}_{j}$ assigns to $\phi(E).$  This discussion shows that for any $0\le\delta,\gamma<1,$
the set of steady states with $\hat{g}$ in $\hat{\mathcal{G}}$ is the
same as the set of steady states with $g$ in $\mathcal{G}$.

Conversely, given a prior density $\hat{g}_{j}$ for every
agent $j$ in the game $\hat{\mathcal{G}}$, we can consider a prior density
$g_{j}$ in $\mathcal{G}$ where $g_{j}(\alpha_{h_{i}^{'}},\alpha_{h_{i}^{''}},(\alpha_{h})_{h\in\mathcal{H}_{-j}\backslash\{h_{i}',h_{i}''\}})=\hat{g}_{j}(\phi(\alpha_{h_{i}^{'}},\alpha_{h_{i}^{''}}),(\alpha_{h})_{h\in\mathcal{\hat{H}}_{-j}\backslash\{h_{i}^{\star}\}})\cdot(\alpha_{h_{i}^{'}}(a_{\text{pass}}))^{n-1}$
for $j\ne i.$ The same argument as above shows for any $0\le\delta,\gamma<1,$
the set of steady states with $g$ in $\mathcal{G}$ is the same as the
set of steady states with $\hat{g}$ in $\hat{\mathcal{G}}$.

\subsection{Patiently Stable Profiles for Figure \ref{fig:doubly_dom} Are Nash Equilibria}\label{sec:auxiliary_game_proof}

\begin{figure}[H]
    \centering
    \includegraphics[scale=.45]{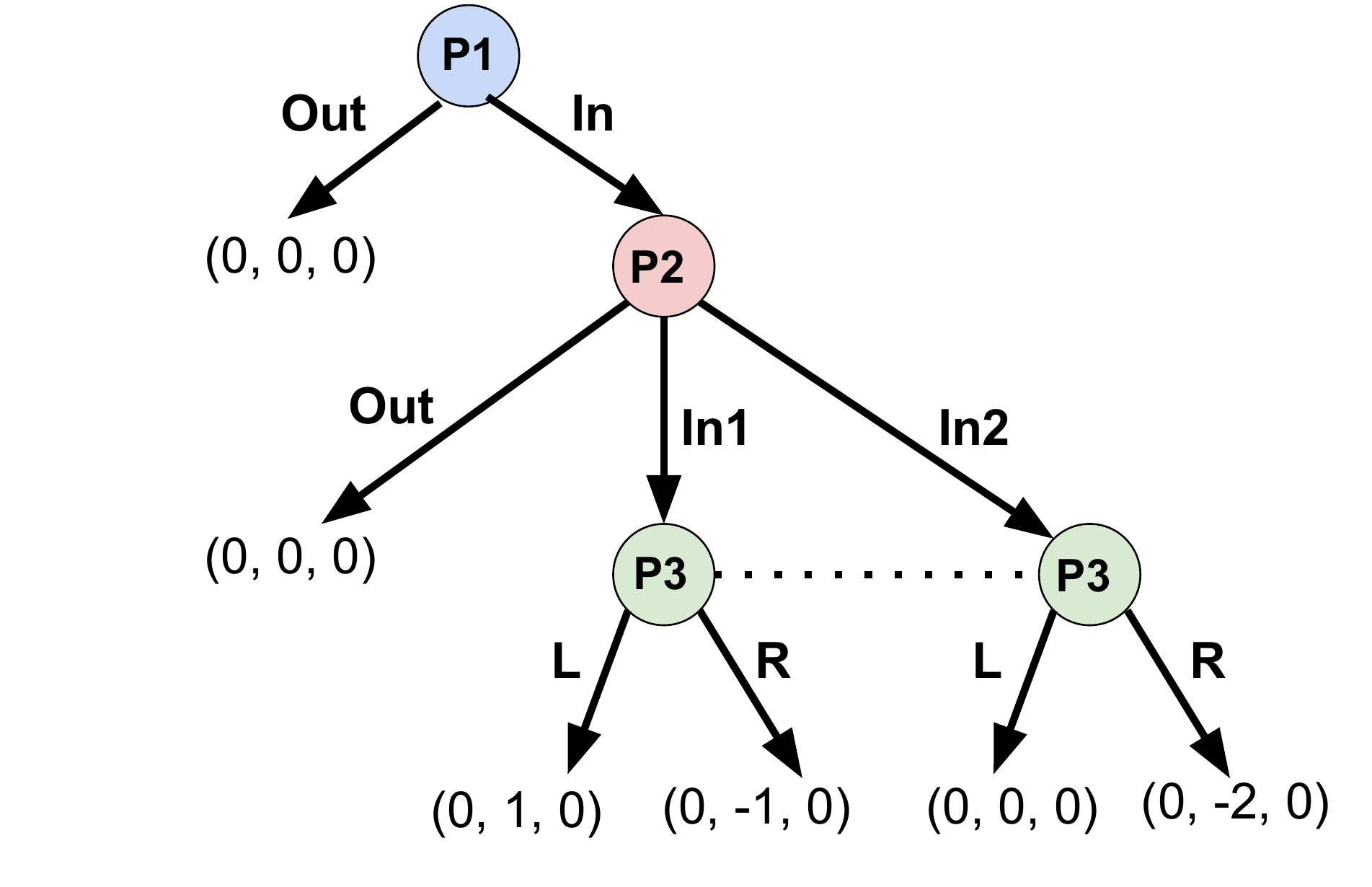}
    \caption{A game with fully revealing terminal node partitions in which the learning problem of a P2 agent with a given prior is identical to the learning problem of a P2 agent with the same prior in the Figure \ref{fig:doubly_dom} game.}
    \label{fig:learning_equivalent_doubly_dom}
\end{figure}
\begin{proof}
Consider the auxiliary game depicted in Figure \ref{fig:learning_equivalent_doubly_dom}, which modifies payoffs to make P1 and P3  indifferent between all terminal nodes, and ends the game immediately if P2 chooses \textbf{Out}. For any $0\le \delta, \gamma < 1$, any P1 or P3 policy used in the original steady state is an optimal policy for the corresponding agent in the auxiliary game, so any steady state profile $\pi^* \in \Pi^*(g,\delta,\gamma)$ for the original game is also a steady state profile of the auxiliary game. And P2 faces the same learning problem in the auxiliary game as in the original game, since each strategy profile gives them the same payoff and the same feedback in both games. But we know in every patiently stable profile in the auxiliary game, P2 must not have a profitable deviation, so the same must apply to the patiently stable profiles of the original game.  Likewise, we can construct auxiliary games for P1 and P3 to show that they must not have profitable deviations in patiently stable profiles of the original game.
\end{proof}


\subsection{Examples of Games with Terminal Node Partitions That Have Auxiliary
Games with Discrete Partitions}\label{sec:example_auxiliary}

The argument in the previous section can be used to show that patient stability selects only Nash equilibria whenever the game and feedback structure are such that, for each player role, there is an auxiliary game with a discrete terminal node partition that leads to the same learning problem for agents in that role. We give some examples of games from the previous literature that meet this condition below.

In \citet{Fudenberg2015}, Figure 1 Game B and Figure 3 present two games where three
players P1, P2, and P3 simultaneously choose actions, P2 and P3 always
see the terminal node, and P1 sees the terminal node when they choose
\textbf{In} but not when they choose \textbf{Out}. P1 always gets
0 payoff from choosing \textbf{Out}. For P2 and P3, consider the auxiliary
game where every player always observes the terminal node. This clearly
does not affect P2 and P3's learning problems. For P1, consider an
auxiliary game where P1 moves first and the game ends with P1 getting
0 payoff if P1 chooses \textbf{Out}. If P1 chooses \textbf{In}, then
P2 and P3 choose their actions simultaneously as before. All players
observe terminal nodes. P1's learning problem in the auxiliary game
is the same as in the original game. Thus for Figure
1 Game B and Figure 
3 with their original terminal node partitions,
patiently stable profiles are Nash equilibria.

Figure 5 of  \citet{Fudenberg2015} is a three-player game where P1, P2 and P3 simultaneously
choose \textbf{In} or \textbf{Out}. When P1 chooses \textbf{In}, they
learn P2 and P3's choices, but P1 does not learn how others play if
they choose \textbf{Out}. P2 always learns how P1 plays, but they
only learn how P3 plays if they choose \textbf{In} rather than \textbf{Out}.
Similarly, P3 always learns how P1 plays, but they only learn how
P2 plays if they choose \textbf{In} rather than \textbf{Out}. Players
who choose \textbf{Out} always get 0. For P1, consider the auxiliary
game where they move first and choose \textbf{In} or \textbf{Out}.
If they choose \textbf{Out}, the game ends with them getting 0. If
they choose \textbf{In}, then P2 and P3 simultaneously choose \textbf{In}
or \textbf{Ou}t. All players observe terminal nodes. This is the same
learning problem as in the original game for P1. Next, consider the
auxiliary game where P1 and P2 move simultaneously at the start of
the game. If P2 chooses \textbf{Out}, the game ends with P2 getting
0. If P2 chooses \textbf{In}, then P3 chooses between \textbf{In}
or \textbf{Out} without knowing P1's choice. All players observe terminal
nodes. This is the same learning problem as in the original game for
P2, because the terminal node always reveals P1's play, even when
P2 chooses \textbf{Out}. But if P2 chooses \textbf{Out}, then the
terminal node does not show what P3 would have played. Similarly,
we can construct an analogous auxiliary game for P3. This shows that
for the game in Figure 5, patiently stable profiles are Nash equilibria.

Section 5.1.1 of \citet{fudenberg2021player} studies the ``restaurant game'' where three
players P1, P2, and P3 move simultaneously. P1 is a restaurant that
chooses between \textbf{high} and \textbf{low} ingredient qualities,
while P2 and P3 are two potential customers who decide whether to
go to the restaurant (\textbf{In}) or eat at home (\textbf{Out}).
P1 always sees P2 and P3's choices. P2 sees how others play if they
choose \textbf{In}, but not if they choose \textbf{Out}. Similarly,
P3 sees how others play if they choose \textbf{In}, but not if they
choose \textbf{Out}. Choosing \textbf{Out} always gives 0 payoff.
For P1, the auxiliary game where everyone sees the terminal node does
not affect their learning problem. For P2, consider the auxiliary
game where they move first, choosing between \textbf{In} and \textbf{Out}.
If they choose \textbf{Out}, the game ends with payoff 0 for them.
If they choose \textbf{In}, then P1 and P3 move simultaneously. All
players observe terminal nodes. This auxiliary game presents the same
learning problem for P2 as in the original game. Similarly, there
is an analogous auxiliary game with discrete terminal node partitions
that preserves P3's learning problem, so patiently stable profiles
are Nash equilibria in this game. 
\subsection{Proof of Proposition \ref{Dominated Action Stability General Result}}
\label{Proof of Proposition Dominated Action Stability General Result}

We first formally define supportive priors that are used to facilitate the proof.

\begin{definition}
Priors $g_{1}$ and $g_{2}$ are \textbf{\emph{supportive}} priors for $\pi^{*}$ if, for every off-path P2 information set $h_{2}$, (1) $\mathbb{E}_{g_{1}}[u_{1}(a_{1}^{*}(h_{2}),a_{2}(h_{2}))|y_{1}] \geq \mathbb{E}_{g_{1}}[u_{1}(a_{1},a_{2}(h_{2}))|y_{1}]$ for all $a_{1} \in \rho(h_{2})$ and P1 histories $y_{1}$ that have never recorded a P2 agent play some action other than $a_{2}^{*}(h_{2})$ at $h_{2}$,
and (2) $\mathbb{E}_{g_{2}}[u_{2}(a_{1},a_{2}^{*}(h_{2}))|y_{2},h_{2}] \geq \mathbb{E}_{g_{2}}[u_{2}(a_{1},a_{2}))|y_{2},h_{2}]$ for all $a_{2} \in \mathcal{A}_{2}(h_{2})$ and histories $y_{2}$ that have never recorded a P1 agent play any $a_{1}\in\rho(h_{2})\backslash\{a_{1}^{*}(h_{2})\}$. 
\label{Supportive Prior Definition 1}
\end{definition}

\begin{proof}[Proof of Proposition \ref{Dominated Action Stability General Result}]

Throughout this proof, we think of P1 actions that end the game as leading to singleton P2 information sets where P2 only has one action. Denote the P2 information set reached when P1 plays $a_{1}^{*}$ with $h_{2}^{*}$, and use $\mathcal{H}_{2}^{off} = \mathcal{H}_{2} \setminus \{h_{2}^{*}\}$ to denote the set of P2 information sets that are off-path under $\pi^{*}$. Let $\widetilde{\Pi}_{1} = \{\pi_{1} \in \Pi_{1} : \forall h_{2} \in \mathcal{H}_{2}^{off}, ~ \pi_{1}(a_{1}) = 0 ~ \forall a_{1} \in \rho(h_{2}) \setminus \{a_{1}^{*}(h_{2})\}\}$ be the set of P1 behavior strategies that, for every $h_{2} \in \mathcal{H}_{2}^{off}$, put probability $0$ on any action in $\rho(h_{2})$ that is not $a_{1}^{*}(\alpha_{1})$. Further, let $\widetilde{\Pi}_{2} = \{\pi_{2} \in \Pi_{2} : \forall h_{2} \in \mathcal{H}_{2}^{off}, ~ \pi_{2}(a_{2}^{*}(h_{2})|h_{2}) = 1\}$ be the set of P2 behavior strategies which respond with $a_{2}^{*}(h_{2})$ at any off-path information set $h_{2}$. Throughout the proof, we restrict attention to strategy profiles $\pi \in \widetilde{\Pi}_{1} \times \widetilde{\Pi}_{2}$.

By continuity, there is an $\eta > 0$ such that (1) for any $\pi_{1} \in \Pi_{1}$ satisfying $\pi_{1}(a_{1}^{*}) \geq 1 - \eta$, the unique optimal action for P2 to play at $h_{2}^{*}$ is $a_{2}^{*}(h_{2}^{*})$, and (2) for any $\pi_{2} \in \widetilde{\Pi}_{2}$ for which $\pi_{2}(a_{2}^{*}(h_{2}^{*})|h_{2}^{*}) \geq 1 - \eta$, the unique P1 best response is $a_{1}^{*}$. We focus on steady state profiles in which the aggregate probabilities that P1 plays $a_{1}^{*}$ and that P2 plays $a_{2}^{*}(h_{2}^{*})$ at $h_{2}^{*}$ both exceed $1 - \eta$. We argue that such steady state profiles exist in the limit, and that the corresponding aggregate probabilities that P1 plays $a_{1}^{*}$ and P2 plays $a_{2}^{*}(h_{2})$ in response to any information set $h_{2}$ converge to $1$.

Let $\xi: \widetilde{\Pi}_{1} \rightarrow \widetilde{\Pi}_{1}$ be the continuous mapping given by
\begin{equation*}
\xi(\pi_{1})(a_{1}) =
\begin{cases}
\max\{\pi_{1}(a_{1}^{*}),1 - \eta\} & \text{if } a_{1} = a_{1}^{*} \\
\left(1 - \mathbbm{1}(\pi_{1}(a_{1}^{*}) < 1 - \eta) \frac{1 - \eta - \pi_{1}(a_{1}^{*})}{1 - \pi_{1}(a_{1}^{*})}\right) \pi_{1}(a_{1}) & \text{if } a_{1} \neq a_{1}^{*}
\end{cases} .
\end{equation*}
This function transforms each $\pi_{1}$ into a P1 behavior strategy that puts probability at least $1 - \eta$ on $a_{1}^{*}$ and satisfies $\xi(\pi_{1}) = \pi_{1}$ whenever $\pi_{1}(a_{1}^{*}) \geq 1 - \eta$. Similarly, let $\phi: \widetilde{\Pi}_{2} \rightarrow \widetilde{\Pi}_{2}$ be the continuous mapping such that
\begin{equation*}
\phi(\pi_{2})(a_{2}|h_{2}^{*}) =
\begin{cases}
\max\{\pi_{2}(a_{2}^{*}(h_{2}^{*})|h_{2}^{*}),1 - \eta\} & \text{if } a_{2} = a_{2}^{*}(h_{2}^{*}) \\
\left(1 - \mathbbm{1}(\pi_{2}(a_{2}^{*}(h_{2}^{*})|h_{2}^{*}) < 1 - \eta) \frac{1 - \eta - \pi_{2}(a_{2}^{*}(h_{2}^{*})|h_{2}^{*})}{1 - \pi_{2}(a_{2}^{*}(h_{2}^{*})|h_{2}^{*})}\right) \pi_{2}(a_{2}|h_{2}^{*}) & \text{if } a_{2} \neq a_{2}^{*}(h_{2}^{*})
\end{cases} .
\end{equation*}
This takes each $\pi_{2}$ into a P2 behavior strategy that uses $a_{2}^{*}(h_{2}^{*})$ at $h_{2}^{*}$ with probability at least $1 - \eta$. Note that $\phi$ coincides with the identity mapping whenever $\pi_{2}(a_{2}^{*}(h_{2}^{*})|h_{2}^{*}) \geq 1 - \eta$.

Since $g_{1}$ is supportive of $\pi^{*}$, for any $\pi_{2} \in \widetilde{\Pi}_{2}$, $\mathscr{R}_{1}^{\delta,\gamma} (\pi_{2})(a_{1}) = 0$ for all $a_{1} \in \rho(h_{2})$ for all $h_{2} \in \mathcal{H}_{2}^{off}$. This means that $\mathscr{R}_{1}^{\delta,\gamma}(\pi_{2}) \in \widetilde{\Pi}_{1}$ for all $\pi_{2} \in \widetilde{\Pi}_{2}$. Likewise, since $g_{2}$ is supportive of $\pi^{*}$, for any $\pi_{1} \in \widetilde{\Pi}_{1}$, $\mathscr{R}_{2}^{\delta, \gamma}(\pi_{1})(a_{2}|h_{2}) = 0$ for all $a_{2} \neq a_{2}^{*}(h_{2})$ for all $h_{2} \in \mathcal{H}_{2}^{off}$. Thus, $\mathscr{R}_{2}^{\delta, \gamma}(\pi_{1}) \in \widetilde{\Pi}_{2}$ for all $\pi_{1} \in \widetilde{\Pi}_{1}$. Consequently, $\mathscr{R}^{\delta,\gamma}$ maps $\widetilde{\Pi}_{1} \times \widetilde{\Pi}_{2}$ into itself regardless of $\delta,\gamma \in [0,1)$, so the mapping $\widetilde{\mathscr{R}}^{\delta,\gamma}: \widetilde{\Pi}_{1} \times \widetilde{\Pi}_{2} \rightarrow \widetilde{\Pi}_{1} \times \widetilde{\Pi}_{2}$ given by $\widetilde{\mathscr{R}}^{\delta,\gamma}(\pi_{1},\pi_{2}) = (\xi(\mathscr{R}_{1}^{\delta,\gamma}(\pi_{2})), \phi(\mathscr{R}_{2}^{\delta, \gamma}(\pi_{1})))$ is well-defined. Since this mapping is continuous, Brouwer's fixed point theorem guarantees the existence of a fixed point $\pi^{\delta,\gamma} = (\pi_{1}^{\delta,\gamma},\pi_{2}^{\delta,\gamma})$ for any $\delta,\gamma \in [0,1)$.

Consider some arbitrary collection of parameter sequences $\{\delta_{j}\}_{j \in \mathbb{N}}$, $\{\gamma_{j,k}\}_{j,k \in \mathbb{N}}$ such that $\lim_{j \rightarrow \infty} \delta_{j} = 1$, $ \lim_{k \rightarrow \infty} \gamma_{j,k} = 1$ for all $j \in \mathbb{N}$, and $\lim_{j \rightarrow \infty} \lim_{k \rightarrow \infty} \pi^{\delta_{j},\gamma_{j,k}} = \hat{\pi}$ for some $\hat{\pi} \in \widetilde{\Pi}_{1} \times \widetilde{\Pi}_{2}$. Using the fact that the unique P1 best response is $a_{1}^{*}$ to any $\pi_{2} \in \widetilde{\Pi}_{2}$, $\lim_{j \rightarrow \infty} \lim_{k \rightarrow \infty} \mathscr{R}_{1}^{\delta_{j},\gamma_{j,k}}(\pi_{2}^{\delta_{j},\gamma_{j,k}})(a_{1}^{*}) = 1$ can be shown using a similar auxiliary game argument to the one given when arguing that patient stability selects Nash equilibria in the Figure \ref{fig:doubly_dom} game. Thus, $\lim_{j \rightarrow \infty} \lim_{k \rightarrow \infty} \mathscr{R}_{1}^{\delta_{j},\gamma_{j,k}}(\pi_{2}^{\delta_{j},\gamma_{j,k}}) =  \pi_{1}^{*}$. As $\xi(\pi_{1}) = \pi_{1}$ if $\pi_{1}(a_{1}^{*}) \geq 1 - \eta$, it follows that $\mathscr{R}_{1}^{\delta_{j},\gamma_{j,k}}(\pi_{2}^{\delta_{j},\gamma_{j,k}}) = \pi_{1}^{\delta_{j},\gamma_{j,k}}$ whenever $j$ is sufficiently large and $k$ is sufficiently large given $j$. Similarly, since $a_{2}^{*}(h_{2}^{*})$ is the uniquely optimal action to use at $h_{2}^{*}$ given any $\pi_{1} \in \widetilde{\Pi}_{1}$, $\lim_{j \rightarrow \infty} \lim_{k \rightarrow \infty} \mathscr{R}_{2}^{\gamma_{j,k}}(\pi_{1}^{\delta_{j},\gamma_{j,k}})(a_{2}^{*}(h_{2}^{*})|h_{2}^{*}) = 1$ must hold. This means that $\lim_{j \rightarrow \infty} \lim_{k \rightarrow \infty} \mathscr{R}_{2}^{\gamma_{j,k}}(\pi_{1}^{\delta_{j},\gamma_{j,k}}) = \pi_{2}^{*}$. Moreover, as $\phi(\pi_{2}) = \pi_{2}$ if $\pi_{2}(a_{2}^{*}(h_{2}^{*})|h_{2}^{*}) \geq 1 - \eta$, it follows that $\mathscr{R}_{2}^{\delta_{j},\gamma_{j,k}}(\pi_{1}^{\delta_{j},\gamma_{j,k}}) = \pi_{2}^{\delta_{j},\gamma_{j,k}}$ whenever $j$ is sufficiently large and $k$ is sufficiently large given $j$. Collecting these findings reveals that $\pi^{\delta_{j},\gamma_{j,k}}$ is a fixed point of the aggregate response mapping $\mathscr{R}^{\delta_{j},\gamma_{j,k}}$, and thus a steady state profile by Proposition \ref{Steady State Existence Result}, whenever $j$ is sufficiently large and $k$ is sufficiently large given $j$. Since $\lim_{j \rightarrow \infty} \lim_{k \rightarrow \infty} \pi^{\delta_{j},\gamma_{j,k}} = \pi^{*}$, we conclude that $\pi^{*}$ is stable. \end{proof}

\subsection{Proof of Proposition \ref{Relatively Experienced P1 Stability Result}}
\label{Proof of Proposition Relatively Experienced P1 Stability Result}

We first establish three lemmas. 

\begin{lemma}
Fix $\delta \in [0,1)$ and a non-doctrinaire P1 prior $g_{1}$. For each $\gamma \in [0,1)$, fix a P1 policy that is optimal given $g_{1}$ and never prescribes \textbf{In} after it has previously prescribed \textbf{Out}. There is some $\kappa \in \mathbb{R}_{+}$ such that, for arbitrary $\gamma \in [0,1)$, when the aggregate P3 strategy puts probability $\pi_{3}^{\delta,\gamma}(\textbf{L}) \leq 1/4$ on $\textbf{L}$, the aggregate P1 strategy satisfies $\pi_{1}^{\delta,\gamma}(\textbf{In}) / (1 - \gamma) \leq \kappa$.
\label{P1 In Probability Bound Result}
\end{lemma}

\begin{proof}

We first establish that there is some $N \in \mathbb{N}$ such that all P1 agents who have lived at $t \geq N$ periods and have been matched with P3 agents that would play \textbf{L} fewer than $t/3$ periods would play \textbf{Out}. By Theorem 4.2 of \citet{Diaconis1990}, there is an $N \in \mathbb{N}$ such that a P1 agent who has played \textbf{In} at least $N$ times and for whom, when they have played \textbf{In}, the share of times they have observed their P3 opponent play \textbf{R} is at least $2/3$,  will put probability at least $3 / (4 - \delta)$ on the true probability with which a randomly selected P3 agent plays \textbf{R} being weakly more than $2/3$. Such an agent thus puts at least probability $3 / (4 - \delta)$ on aggregate opponent behavior strategy profiles for which the expected payoff from playing \textbf{In} is no more than $-1/3$. This $N$ satisfies the desired properties given at the beginning of the paragraph. To see this, consider a P1 agent who has lived at least $N$ periods and for whom the fraction of time periods where they were matched with a P3 agent that would play \textbf{L} that periods is less than $1/3$. Then, either that agent has played \textbf{Out} in the past, in which case they will again play \textbf{Out}, or that agent has actually observed the consequences of playing \textbf{In} at least $N$ times. Restricting attention to the latter case, the agent must have a posterior belief that puts probability $p > 3 / (4 - \delta)$ on aggregate opponent behavior strategy profiles for which the expected payoff from playing \textbf{In} is no more than $-1/3$. An upper bound on the agent's expected discounted future lifetime payoff from playing \textbf{In} is $(1 - \delta) (1 - 4 p / 3) + \delta (1 - p)$, since the expected current period payoff to playing \textbf{In} is weakly less than $p(-1/3) + 1 - p = 1 - 4 p / 3$ and the agent's continuation payoff is bounded above by $\delta(1 - p)$, since P1's maximum payoff is $1$. As $p > 3 / (4 - \delta)$, it follows that $(1 - \delta) (1 - 4 p / 3) + \delta (1 - p) < 0$, so such an agent must play \textbf{Out}.

We now combine this fact with Hoeffding's inequality to derive the desired constraint on the P1 aggregate strategy. By Hoeffding's inequality, there is some $c > 0$ such that, for any aggregate P3 strategy satisfying $\pi_{3}^{\delta,\gamma}(\textbf{L}) \leq 1/4$, the share of P1 agents who have lived $n$ periods and for whom the fraction of time periods where they were matched with a P3 agent that would play \textbf{L} is less than $1/3$ is at least $1 - e^{- c n}$. Thus, we have that
\begin{equation*}
\begin{split}
\frac{\pi_{1}^{\delta,\gamma}(\textbf{In})}{1 - \gamma} & \leq \frac{1}{1 - \gamma} \left(1 - \gamma^{N} + \sum_{t = N}^{\infty} (1 - \gamma) \gamma^{t} e^{- c t}\right) \\
& = \frac{1 - \gamma^{N}}{1 - \gamma} + \frac{\gamma^{N} e^{- c N}}{1 - \gamma e^{-c}} .
\end{split}
\end{equation*}
Observe that the right-hand side of the inequality converges to $N + 1 / (e^{c N} - e^{c (N - 1)})$ as $\gamma \rightarrow 1$. Since $\pi_{1}^{\delta,\gamma}(\textbf{In}) / (1 - \gamma)$ can never be more than $1 / (1 - \gamma)$, it follows that $\pi_{1}^{\delta,\gamma}(\textbf{In}) / (1 - \gamma)$ must be uniformly bounded from above by some $\kappa \in \mathbb{R}_{+}$.
\end{proof}

\begin{lemma}
Fix $\delta \in [0,1)$ and a non-doctrinaire P2 prior $g_{2}$ under which the expected probability of $\textbf{L}$ is strictly less than $1/2$. Consider a sequence of steady states such that the probability of \textbf{In} under the aggregate P1 strategy satisfies $\pi_{1}^{\delta,\gamma}(\textbf{In}) / (1 - \gamma) \leq \kappa$ for all $\gamma \in [0,1)$ for some $\kappa \in \mathbb{R}_{+}$. Then $\lim_{\gamma \rightarrow 1} \pi_{2}^{\delta,\gamma}(\textbf{Out}) = 1$.
\label{P2 Out Probability Bound Result}
\end{lemma}

\begin{proof}

Fix an $\epsilon > 0$ such that the expected probability of $\textbf{L}$ under $g_{2}$ is weakly less than $(1 - \epsilon) / 2$. We first establish that there is some $N_{0} \in \mathbb{N}$ such that, regardless of $\gamma \in [0,1)$, every P2 agent who has lived at least $N_{0}$ periods and has never observed a P1 agent play \textbf{In} will play \textbf{Out}. Theorem 4.2 of \citet{Diaconis1990} implies that there is an $N \in \mathbb{N}$ such that, under the posterior belief of a P2 agent who has at least $N$ observations of P1 agents playing \textbf{Out} and no observations of a P1 agent playing \textbf{In}, the expected value of the probability with which a randomly selected P1 agent will play \textbf{In}, $\nu$, is strictly less than $(1 - \delta) \epsilon / \delta$. This $N$ satisfies the desired properties given at the beginning of the paragraph. To see this, consider a P2 agent who has lived at least $N$ periods and has never observed a P1 agent play \textbf{In}. An upper bound on the agent's expected discounted future lifetime payoff from playing \textbf{In1} is $- (1 - \delta) \epsilon + \delta \nu$, which is strictly negative since $\nu < (1 - \delta) \epsilon / \delta$, so such an agent must play \textbf{Out}.

Inductively applying similar arguments to the one given above shows that there is a sequence of $\{N_{j}\}_{j \in \mathbb{N}} \subseteq \mathbb{N}$ such that the following holds. For all $\gamma \in [0,1)$ and $j \in \mathbb{N}$, every P2 agent who has lived at least $N_{j}$ periods, has at most $j$ observations of P1 agents playing \textbf{In}, and witnessed no P1 agents playing \textbf{In} in their first $N_{j}$ observations will play \textbf{Out}. Observe that, when the probability of a randomly selected P1 agent playing \textbf{In} is $\pi_{1}(\textbf{In})$, the share of P2 agents who have lived at least $N_{j} + j$ periods and have exactly $j$ observations of P1 agents playing \textbf{In}, all of which came after their first $N_{j}$ periods, is
\begin{equation*}
\begin{split}
& \sum_{t = N_{j} + j}^{\infty} (1 - \gamma) \gamma^{t} \frac{(t - N_{j})!}{(t - N_{j} - j)! j!} \pi_{1}^{\delta,\gamma}(\textbf{In})^{j} \left(1 - \pi_{1}^{\delta,\gamma}(\textbf{In})\right)^{t - j} \\
= & \gamma^{N_{j} + j} \left(1 - \pi_{1}^{\delta,\gamma}(\textbf{In})\right)^{N_{j}} \frac{\pi_{1}^{\delta,\gamma}(\textbf{In})^{j}}{\left(1 - \gamma + \gamma \pi_{1}^{\delta,\gamma}(\textbf{In})\right)^{j + 1}} \\
= & \gamma^{N_{j} + j} \left(1 - \pi_{1}^{\delta,\gamma}(\textbf{In})\right)^{N_{j}} \left(\frac{1 + \frac{\pi_{1}^{\delta,\gamma}(\textbf{In})}{1 - \gamma}}{1 + \gamma\frac{\pi_{1}^{\delta,\gamma}(\textbf{In})}{1 - \gamma}}\right)^{j + 1} \frac{\left(\frac{\pi_{1}^{\delta,\gamma}(\textbf{In})}{1 - \gamma}\right)^{j}}{\left(1 + \frac{\pi_{1}^{\delta,\gamma}(\textbf{In})}{1 - \gamma}\right)^{j + 1}} .
\end{split}
\end{equation*}
It thus follows that, for a given $\gamma \in [0,1)$ and steady-state strategy profile,
\begin{equation*}
\begin{split}
\pi_{2}^{\delta,\gamma}(\textbf{Out}) & \geq \sum_{j = 0}^{\infty} \gamma^{N_{j} + j} \left(1 - \pi_{1}^{\delta,\gamma}(\textbf{In})\right)^{N_{j}} \left(\frac{1 + \frac{\pi_{1}^{\delta,\gamma}(\textbf{In})}{1 - \gamma}}{1 + \gamma\frac{\pi_{1}^{\delta,\gamma}(\textbf{In})}{1 - \gamma}}\right)^{j + 1} \frac{\left(\frac{\pi_{1}^{\delta,\gamma}(\textbf{In})}{1 - \gamma}\right)^{j}}{\left(1 + \frac{\pi_{1}^{\delta,\gamma}(\textbf{In})}{1 - \gamma}\right)^{j + 1}} \\
& = 1 - \sum_{j = 0}^{\infty} \left(1 - \gamma^{N_{j} + j} \left(1 - \pi_{1}^{\delta,\gamma}(\textbf{In})\right)^{N_{j}} \left(\frac{1 + \frac{\pi_{1}^{\delta,\gamma}(\textbf{In})}{1 - \gamma}}{1 + \gamma \frac{\pi_{1}^{\delta,\gamma}(\textbf{In})}{1 - \gamma}}\right)^{j + 1} \right) \frac{\left(\frac{\pi_{1}^{\delta,\gamma}(\textbf{In})}{1 - \gamma}\right)^{j}}{\left(1 + \frac{\pi_{1}^{\delta,\gamma}(\textbf{In})}{1 - \gamma}\right)^{j + 1}} \\
& \geq 1 - \sum_{j = K + 1}^{\infty} \frac{\left(\frac{\pi_{1}^{\delta,\gamma}(\textbf{In})}{1 - \gamma}\right)^{j}}{\left(1 + \frac{\pi_{1}^{\delta,\gamma}(\textbf{In})}{1 - \gamma}\right)^{j + 1}} \\
& ~~ - \sup \left\{1 - \gamma^{N_{j} + j} \left(1 - \pi_{1}^{\delta,\gamma}(\textbf{In})\right)^{N_{j}} \left(\frac{1 + \frac{\pi_{1}^{\delta,\gamma}(\textbf{In})}{1 - \gamma}}{1 + \gamma \frac{\pi_{1}^{\delta,\gamma}(\textbf{In})}{1 - \gamma}}\right)^{j + 1}\right\}_{j \in \{0,1,...,K\}}
\end{split}
\end{equation*}
for arbitrary $K \in \mathbb{N}$. Fix an arbitrary $\eta > 0$ and take $K$ to be large enough so that $\sum_{j = K + 1}^{\infty} \kappa^{j} / (1 + \kappa)^{j + 1} < \eta$. Then the right-hand side of the first line of the final inequality is greater than $1 - \eta$ for all $\pi_{1}^{\delta,\gamma}(\textbf{In}) / (1 - \gamma) \leq \kappa$. Observe that the elements of the set over which the supremum is taken in the final line converge to $0$ as $\gamma \rightarrow 1$ uniformly over $\pi_{1}^{\delta,\gamma}(\textbf{In}) / (1 - \gamma) \leq \kappa$. Thus, $\liminf_{\gamma \rightarrow 1} \pi_{2}^{\delta,\gamma}(\textbf{Out}) \geq 1 - \eta$. Since this holds for all $\eta > 0$, we have $\lim_{\gamma \rightarrow 1} \pi_{2}^{\delta,\gamma}(\textbf{Out}) = 1$. \end{proof}

\begin{lemma}
Fix $\delta \in [0,1)$ and a non-doctrinaire P3 prior $g_{3}$ that leads a P3 agent to only play \textbf{L} when they have previously observed a P2 agent play \textbf{In1}. Consider a sequence of steady-states such that the probability of \textbf{In} under the aggregate P1 strategy satisfies $\pi_{1}^{\delta,\gamma}(\textbf{In}) / (1 - \gamma) \leq \kappa$ for all $\gamma \in [0,1)$ for some $\kappa \in \mathbb{R}_{+}$ and $\lim_{\gamma \rightarrow 1} \pi_{2}^{\delta,\gamma}(\textbf{In1}) = 0$. Then $\lim_{\gamma \rightarrow 1} \pi_{3}^{\delta,\gamma}(\textbf{L}) = 0$.
\label{P3 L Probability Bound Result}
\end{lemma}

\begin{proof}

Observe that
\begin{equation*}
\begin{split}
\pi_{3}^{\delta,\gamma}(\textbf{L}) & \leq \sum_{t = 0}^{\infty} (1 - \gamma) \gamma^{t} (1 - (1 - \pi_{1}^{\delta,\gamma}(\textbf{In}) \pi_{2}^{\delta,\gamma}(\textbf{In1}))^{t}) \\
& = 1 - \frac{1 - \gamma}{1 - \gamma (1 - \pi_{1}^{\delta,\gamma}(\textbf{In}) \pi_{2}^{\delta,\gamma}(\textbf{In1}))} \\
& = 1 - \frac{1}{1 + \gamma \frac{\pi_{1}^{\delta,\gamma}(\textbf{In})}{1 - \gamma} \pi_{2}^{\delta,\gamma}(\textbf{In1})}
\end{split}
\end{equation*}
since the right-hand side of the inequality is the share of P3 agents who have previously observed a P2 agent play \textbf{In1}. By Lemma \ref{P1 In Probability Bound Result}, there is some $\kappa \in \mathbb{R}_{+}$ such that $\limsup_{\gamma \rightarrow 1} \pi_{1}^{\delta,\gamma} / (1 - \gamma) \leq \kappa$, while $\lim_{\gamma \rightarrow 1} \pi_{2}^{\delta,\gamma}(\textbf{In1}) = 0$ by Lemma \ref{P2 Out Probability Bound Result}. It follows that $\lim_{\gamma \rightarrow 1} 1 + \gamma (\pi_{1}^{\delta,\gamma}(\textbf{In})) / (1 - \gamma) \pi_{2}^{\delta,\gamma}(\textbf{In1}) = 1$, which implies $\lim_{\gamma \rightarrow 1} \pi_{3}^{\delta,\gamma}(\textbf{L}) = 0$. \end{proof}

\begin{proof}[Proof of Proposition \ref{Relatively Experienced P1 Stability Result}]

Lemmas \ref{P1 In Probability Bound Result}, \ref{P2 Out Probability Bound Result}, and \ref{P3 L Probability Bound Result} together imply that, for fixed $\delta \in [0,1)$, the aggregate response mapping maps the set of aggregate strategy profiles where $\pi_{3}^{\delta,\gamma}(L) \leq 1/4$ into itself when $\gamma$ is close enough to $1$. Brouwer's fixed point theorem then guarantees the existence of a steady state profile satisfying this inequality for all sufficiently high $\gamma$. Lemmas \ref{P1 In Probability Bound Result}, \ref{P2 Out Probability Bound Result}, and \ref{P3 L Probability Bound Result} further imply that, in the $\gamma \rightarrow 1$ limit of such a sequence of steady state profiles, $\lim_{\gamma \rightarrow 1} \pi_{1}^{\delta,\gamma}(\textbf{Out}) = 1$, $\lim_{\gamma \rightarrow 1} \pi_{2}^{\delta,\gamma}(\textbf{Out}) = 1$, and $\lim_{\gamma \rightarrow 1} \pi_{3}^{\delta,\gamma}(\textbf{R}) = 1$ must be satisfied. Since $\delta \in [0,1)$ is arbitrary, we conclude that $(\textbf{Out},\textbf{Out},\textbf{In})$ is patiently stable. \end{proof}

\subsection{Proof of Proposition \ref{prop:normal_form_elimination}}
\label{Proof of Proposition prop:normal_form_elimination}

We first state a supporting lemma that shows that, with enough data, a given agent's posterior beliefs will, for every opponent population, put high probability on the empirical distribution of strategies they have previously observed agents in that population use. 

\begin{lemma}
For any fixed non-doctrinaire prior and every $\eta > 0$, there is some $M$ such that, whenever an agent has at least $M$ observations, for each opponent population, the agent's posterior belief puts probability $1 - \eta$ on strategy distributions within $\eta$, under the sup norm, of the empirical distribution they have observed.
\label{Adaptation of Diaconis-Freedman}
\end{lemma}

This follows from the \citet{pathwise} extension of the pathwise concentration result of \citet{Diaconis1990} to priors that do not have full support. The support restriction arise because the agent's prior is concentrated on distributions that can be generated by independent randomizations of their opponents.

\begin{proof}[Proof of Proposition \ref{prop:normal_form_elimination}]

Fix a discount factor $\delta \in [0,1)$, and consider a sequence of survival probabilities $\{\gamma_{k}\}_{k=1}^{\infty}$ with an associated sequence of steady-state strategy profiles $\{\pi_{k}\}_{k=1}^{\infty}$
such that $\pi_{k}\to\pi^{*}$. We show that $\pi_{i}^{*}(s_{i})=0$ for each $i$ and each $s_{i}\notin S_{i}^{*}$.

First, it is clear that $\pi_{i}^{*}(s_{i})=0$ for each $i$ and
each $s_{i}\notin S_{i}^{(0)}$. This is because agents have full-support
posterior beliefs after every history and their observations do not
depend on their play, so they never use weakly dominated strategies.

By Lemma \ref{Adaptation of Diaconis-Freedman}, for a given $\eta > 0$, there is some $M$ such that, whenever a player has at least $M$ observations, their posterior belief over the prevailing strategy distribution in each of their opponent populations puts probability $1 - \eta$ on strategy distributions within $\eta$ of the empirical distribution they have observed. By the law of large numbers, we can choose this $M$ to be such that the posterior beliefs of an agent in an arbitrary player role $i$ who has lived at least $M$ periods will be accurate with high probability in the following sense. With probability $1 - 2 \eta$, at the end of the period, following any possible observation in the period itself, the agent's posterior belief puts probability at least $1 - 2 \eta$ on strategy distributions within $2 \eta$ of the true prevailing distribution for each opponent role $j \neq i$. 

Now suppose inductively that $\pi_{i}^{*}(s_{i})=0$ for each $i$
and $s_{i}\notin S_{i}^{(m)}$ for some $m$. Fix arbitrary $\epsilon, \nu > 0$ and restrict attention to $k$ large enough so that $|\pi_{i,k} - \pi_{i}^{*}| < \epsilon / 2$ for all player roles $i$. By the preceding argument, we know that for all sufficiently large $k$, in a given period the share of player $i$ agents whose posterior beliefs at the end of the period, regardless of their observations during the period, for each opponent role $j$ put probability at least $1 - \epsilon$ on strategy distributions within $1 - \epsilon$ of $\Delta(S_{j}^{(m)})$ will exceed $1 - \nu$. Let $\sigma \in \Delta(S_{-i})$ be the expectation held by such an agent about the play of their opponents in the current period. The properties of the agent's beliefs imply that $\sigma$ is full support and that $\sigma(S_{j}^{(m)}|s_{-ij}) \geq (1 - \epsilon)^{2}$ for all $s_{-ij} \in S_{-ij}$. For all sufficiently small $\epsilon$, every $s_{i} \not\in S_{i}^{(m)}$ must be suboptimal for such an agent, so $\pi_{i}^{*}(s_{i}) \leq \nu$ must hold for each $i$ and $s_{i}\notin S_{i}^{(m + 1)}$. Since this is true for all $\nu > 0$, we conclude that $\pi_{i}^{*}(s_{i}) = 0$ for each $i$ and $s_{i}\notin S_{i}^{(m + 1)}$.



\end{proof}

\subsection{Proof of Proposition \ref{prop:simple_game_BI}}
\label{Proof of Proposition prop:simple_game_BI}

\begin{proof}
Let $D_{i}^{(m)}$ be those extensive-form strategies of $i$ that
choose an action inconsistent with backward induction at
a decision node that is $m+1$ steps away from the terminal nodes,
but do not do so at any decision nodes closer to the terminal nodes.
To see these choices are valid for the iterative procedure, first
note $D_{i}^{(0)}$ are weakly dominated for $i$: For any $s_{i}\in D_{i}^{(0)},$
consider a different strategy $s_{i}'$ that changes one of the non-backward-induction
actions at one of $i$'s decision nodes $h_{i}$ one step away from
terminal nodes to a backward-induction action. Then $u_{i}(s_{i}',s_{-i})\ge u_{i}(s_{i},s_{-i})$
for all $s_{-i}$. Moreover, there exists at least one $s_{-i}^{*}$ such that
$h_{i}$ is reached (since the game is simple), so $u_{i}(s_{i}',s_{-i}^{*})>u_{i}(s_{i},s_{-i}^{*}).$

By definition, $S_{i}^{(m)}$ are the strategies where $i$ uses the backward-induction
action at all decision nodes $m+1$ steps or fewer away from terminal
nodes. To see that each $s_{i}\in D_{i}^{(m+1)}$ fails to be a best
response for $i$ to full-support conjectures of  their opponents' play that put high
conditional probabilities on $S_{j}^{(m)}$ for each $j \neq i$, let $h_{i}$ be a decision node
$m+2$ steps away from terminal nodes where $s_{i}$ does not choose
the backward-induction action, and let $s_{i}'$ be the strategy that only
differs from $s_{i}$ in that $s_{i}'$ by selecting a backward-induction
action at $h_{i}.$ Let $J$ be the set of players who have decision
nodes in the subgame starting at $h_{i}$, not counting $h_{i}$ itself.
Because the game is simple, $i\notin J$. For the same reason,  whether  play reaches $h_{i}$ does not depend on the strategy of $i$ or the strategies of the players in $J$.  
 Let $S_{-iJ}^{reach}\subseteq S_{-iJ}$
be the set of strategies of $-iJ$ that reach $h_{i}$.

For any $s_{-i}\in S_{-i}^{(m)},$ $i$'s payoff in the subgame starting
at $h_{i}$ is strictly higher with $s_{i}'$ than with $s_{i},$
which we write as $u_{i}(s_{i},s_{-i}\mid h_{i})<u_{i}(s_{i}',s_{-i}\mid h_{i}).$
This uses the fact that $i$'s payoff in the subgame starting
at $h_{i}$ only depends on $i$'s action at $h_{i}$ and on the strategies
of players in $J$. Consider any strictly mixed profile $\sigma_{-i}$ where
$\sigma_{-i}(S_{J}^{(m)}|s_{-iJ})\ge1-\epsilon$ for each
$s_{-iJ} \in S_{-iJ}$. When $-iJ$ choose a strategy profile in $S_{-iJ}^{reach}$,
$i$'s payoff is equal to $i$'s payoff in the subgame starting at
$h_{i}.$ When $-iJ$ choose a strategy outside of $S_{-iJ}^{reach}$,
$i$ is indifferent between $s_{i}$ and $s_{i}'.$ Therefore, $u_{i}(s_{i},\sigma_{-i})-u_{i}(s_{i}',\sigma_{-i})=\sigma_{-i}(S_{-iJ}^{reach})\cdot[\mathbb{E}[u_{i}(s_{i},s_{-i}\mid h_{i})|S_{-iJ}^{reach}]-u_{i}(s_{i}',s_{-i}\mid h_{i})|S_{-iJ}^{reach}]].$
We have $\sigma_{-iJ}(S_{-iJ}^{reach})>0$ since each opponent's strategy
is strictly mixed, and we have $u_{i}(s_{i},(\sigma_{j})_{j\in J}\mid h_{i})-u_{i}(s_{i}',(\sigma_{j})_{j\in J}\mid h_{i})<0$
for all sufficiently small $\epsilon > 0$. \end{proof}

\subsection{Proof of Proposition \ref{prop:naive_terminal_nodes}}
\label{Proof of Proposition prop:naive_terminal_nodes}
\begin{proof}
First note that for every non-doctrinaire prior $g$ over behavior
strategies in $\mathcal{G},$ there is a non-doctrinaire prior $\hat{g}$
over mixed strategies in $\mathcal{G}$ that generates the same set of
steady states for every $0\le\delta,\gamma<1$. This is because each
$-i$ behavior strategy $(\alpha_{h_{-i}})_{h_{-i}\in\mathcal{H}_{i}}$
is associated with an equivalent mixed strategy $\sigma_{-i}\in\Delta(\mathbb{S}_{-i}),$
defined by $\sigma_{-i}(s_{-i})=\times_{h_{-i}\in\mathcal{H}_{-i}}\alpha_{h_{-i}}(s_{-i}(h_{-i}))$
for each $s_{-i}\in\mathbb{S}_{-i}$. This association  maps the interior of the set of behavior strategies onto the interior
of the set of mixed strategies, so the non-doctrinaire $g_{i}$ generates
a non-doctrinaire $\hat{g}_{i}$ over $-i$'s mixed strategies. Conversely,
if we start with a non-doctrinaire prior $\hat{g}$ over mixed strategies
in $\mathcal{G}$, then by applying Kuhn's theorem in a game with perfect
recall, we can identify a non-empty set of equivalent behavior strategies
for every mixed strategy. Consider the prior $g_{i}$ over $-i$ behavior
strategies where $i$ believes $-i$ first draw a mixed strategy $\sigma_{-i}$
according to $\hat{g}_{i}$, and then randomize uniformly over all
behavior strategies equivalent to it. Then $g_{i}$ is strictly positive on the interior because $\hat{g}_{i}$ enjoys
the same property.

Learning with a non-doctrinaire prior $\hat{g}$ over mixed strategies
in $\mathcal{G}$ with terminal node partitions $\mathcal{P}$ and learning with the same $\hat{g}$ in $\mathcal{N}$
with the $\mathcal{P}$-equivalent partitions generate the same set
of steady states for every $0\le\delta,\gamma<1$. This
is because both environments generate the same dynamic optimization
problem for each agent: in both environments, they start with the
same prior beliefs, receive the same payoffs for each strategy profile
$(s_{i},s_{-i})$ played, and observe the same information (up to
identifying elements of the $\mathcal{P}_i$ partition with those in the equivalent $\hat{\mathcal{P}}_i$ partition.) 
\end{proof}

\subsection{Proof of Claim \ref{claim:coarser_partition}}
\label{Proof of Claim claim:coarser_partition}

\begin{proof}

Suppose there is a prior $g$ satisfying the hypotheses of the claim,
parameters $\{\delta_{j}\}_{j\in\mathbb{N}},$ $\{\gamma_{j,k}\}_{j,k\in\mathbb{N}}$,
and associated steady-state profiles $\{\pi_{j,k}\in\Pi^{*}(g,\delta_{j},\gamma_{j,k})\}$
such that $\lim_{j\to\infty}\delta_{j}=1,$ $\lim_{k\to\infty}\gamma_{j,k}=1$
for each $j$, and $\lim_{j\to\infty}\lim_{k\to\infty}\pi_{j,k}=\pi$.

For arbitrary $\epsilon>0$, we will show $\pi(\textbf{R})<2\epsilon$.
The idea is to show that all P2 agents except the very young and those
with unusual samples will have seen enough instances of P1 choosing
\textbf{In2} and no instance of P1 choosing \textbf{Out }as to play
\textbf{L} at their information set. 

By Proposition 1 from \citet{fudenberg2017bayesian}, there exists some $x\ge1$ so that
if a P2 agent has $n$ observations of P1's play and in each of their observations P1 never chose
\textbf{In1, }then their mean posterior probability of P1 choosing
\textbf{In1} is lower than $\frac{2x}{n+x}$. By Theorem 1 from \citet{fudenberg2017bayesian}, there exists $N\ge1$ such that in any steady state $\hat{\pi}$
where $\hat{\pi}({\textbf{In2}})\ge q$, with probability
at least $1-\epsilon$ a P2 agent with age at least $N/q$ will have
a mean posterior belief of P1 playing\textbf{ In2} that is higher
than $(1-\epsilon)q$.

Define the constant $K=\frac{16Nx}{\epsilon }$ and find some $\beta<1$
so that, whenever $\delta\ge\beta$ and $\gamma\ge\beta$, a P1
agent will always choose \textbf{In2} in the first $K$ periods of
life.

Consider any $j$ large enough so that 
 $\delta_{j}\ge\beta$. For large $k$, using the fact that P1
agents experiment with \textbf{In2} for at least $K$ periods, $\pi_{j,k}({\textbf{In2}})\ge(1+\gamma_{j,k}+...+\gamma_{j,k}^{K-1})\cdot(1-\gamma_{j,k})\ge\frac{1}{2}(1-\gamma_{j,k})\cdot K$.
So, a P2 agent aged at least $\frac{N}{\frac{1}{2}(1-\gamma_{j,k})\cdot K}=(\frac{\epsilon}{2}\frac{1}{1-\gamma_{j,k}})\cdot\frac{1}{4\cdot x}$
has at least $1-\epsilon$ chance of believing that P1 plays \textbf{In2}
with probability at least $\frac{1}{4}(1-\gamma_{j,k})\cdot K$.
This age is no larger than $\frac{\epsilon}{2}$
times the expected P2 lifespan, which contains at least $1-\epsilon$
fraction of the P2 population.  Also, a P2 agent with age at least $\frac{\epsilon}{2}\frac{1}{1-\gamma_{j,k}}$
has a mean posterior belief about P1 playing \textbf{In1} that is
always smaller than $\frac{2x}{\frac{\epsilon}{2}\frac{1}{1-\gamma_{j,k}}+x}=\frac{2x(1-\gamma_{j,k})}{\frac{1}{2}\epsilon +(1-\gamma_{j,k})x}\le\frac{4x}{\epsilon }(1-\gamma_{j,k})$.
Taking the ratio of the mean posterior beliefs assigned to P1 playing
\textbf{In2} and \textbf{In1}, we get $\frac{\frac{1}{4}(1-\gamma_{j,k})\cdot K}{\frac{4x}{\epsilon }(1-\gamma_{j,k})}=\frac{1}{16}K\cdot\frac{\epsilon}{x}=N\ge1.$ 

Therefore, except for a mass of smaller than $\epsilon$ of P2s younger
than $\frac{\epsilon}{2}\cdot\frac{1}{1-\gamma_{j,k}}$ and another
mass $\epsilon$ of P2s with unusual samples, P2s   respond to
\textbf{In1} and \textbf{In2} with \textbf{L}. This shows in the steady
state with $\delta_{j}$ and $k$ large enough, $\pi_{j,k}(\textbf{R})<2\epsilon$.
This implies also that $\lim_{k\to\infty}\pi_{j,k}(\textbf{R})<2\epsilon$ for all large enough
$j,$ therefore $\pi(\textbf{R})<2\epsilon.$ 
\end{proof}

\end{document}